\documentclass[prx,onecolumn,english,superscriptaddress,floatfix,longbibliography]{revtex4-2}

\usepackage[table]{xcolor}
\usepackage{tabularray}
\usepackage[version=3]{mhchem}
\usepackage{amsfonts}
\usepackage{amsmath}
\usepackage{braket}
\usepackage{hyperref}
\usepackage{tikz-cd}
\usepackage{minted}
\usepackage{epigraph}
\usepackage{enumitem}
\usepackage{systeme}
\usepackage{tcolorbox}
\usepackage{booktabs}
\usepackage{multirow}
\usepackage{graphicx}
\usepackage{rotating}
\usepackage{caption}
\usepackage{amsthm}
\usepackage{amssymb}
\usepackage{standalone}
\usepackage{quantikz}
\usepackage[caption=false]{subfig}

\definecolor{jw}{HTML}{BD7DBD}        
\definecolor{jwcas}{HTML}{4477AA}     
\definecolor{jwcassf}{HTML}{EE6677}   
\definecolor{saecas}{HTML}{FFA500}    

\newcommand{\legendsquare}[1]{\textcolor{#1}{\rule{1.4ex}{1.4ex}}}

\newtheorem{proposition}{Proposition}

\begin{document}

\title{Symmetry-adapted qubit encoding with complete active space and Bravyi--Kitaev mapping for quantum chemistry on a quantum computer}

\author{Dario Picozzi}
\affiliation{Department of Physics and Astronomy, University College London (UCL), Gower Street, London, WC1E 6BT, United Kingdom}
\affiliation{London Centre for Nanotechnology, 19 Gordon St, London, WC1H 0AH, United Kingdom}
\email{picozzi.dario@gmail.com}

\author{Jonathan Tennyson}
\affiliation{Department of Physics and Astronomy, University College London (UCL), Gower Street, London, WC1E 6BT, United Kingdom}

\begin{abstract}
We present a symmetry-adapted qubit encoding with complete active space (SAE-CAS) for quantum chemistry on fault-tolerant and near-term quantum processors. Building on exact-symmetry encodings, we extend symmetry-adapted mappings to approximate $Z$-symmetries corresponding to frozen-core and virtual orbitals, thereby reducing qubit requirements without significant loss of accuracy. We derive the mapping from the second-quantised Hamiltonian to active-space qubit Hamiltonians, prove its equivalence to the canonical CAS Hamiltonian with frozen-core and virtual-orbital projection, and integrate it with point-group and spin-parity symmetry encodings via affine Clifford transformations to maximise qubit reduction while preserving the target symmetry sector. The same framework also accommodates the Bravyi--Kitaev mapping, yielding an SAE-CAS-BK variant that is unitarily equivalent to SAE-CAS. Numerical benchmarking on nine small molecules using UCCSD and a hardware-efficient shifted-circular-alternating (HE-SCA) ansatz shows that SAE-CAS reduces qubit counts and Pauli-operator weight, yields shallower circuits with fewer parameters, and often accelerates VQE convergence; with HE-SCA it consistently reaches CAS reference energies in cases where JW-CAS does not converge within the tested budgets. We provide an open-source implementation in the Python package \texttt{QuantumSymmetry}. SAE-CAS offers a route to resource-efficient molecular simulations on fault-tolerant and near-term quantum processors.
\end{abstract}

\maketitle

\section{Introduction}

In the complete active space (CAS) approximation \cite{Roos1980,Helgaker2000}, molecular orbitals in the electronic structure problem are partitioned into three categories: the frozen-core orbitals are approximated as being fully occupied, the virtual orbitals are approximated as being unoccupied, and the active orbitals are allowed any occupancies. A common notation is ($n$,$m$)-CAS, which indicates that $n$ electrons are distributed in all possible ways in $m$ molecular orbitals.

From the first proposals of quantum computers, quantum chemistry has been identified as a natural application \cite{Feynman1982,Cao2019,McArdle2020QCC}. In recent years, significant effort in the literature has been devoted to studying quantum chemistry on the imperfect quantum computers that we currently have, or noisy intermediate-scale quantum (NISQ) devices \cite{Preskill2018,McClean2016}, as this is seen by many as one of the main applications where quantum computers might achieve quantum advantage \cite{Arute2019}. 

The variational quantum eigensolver (VQE) is a hybrid quantum--classical algorithm that approximates molecular ground-state energies by iteratively optimising a parameterised quantum circuit on the device with a classical optimiser. The VQE has been enhanced by various different proposals throughout the literature \cite{McArdle2019,Endo2020,Benedetti2021,Gacon2023,Gomes2021,Stokes2020,Grimsley2019,Huggins2019,Verteletskyi2020,Tang2021}.

We have previously provided an algorithm to construct symmetry-adapted encodings (SAEs) that reduce the number of qubits in the qubit Hamiltonian by exploiting exact symmetries such as point-group and number-conservation symmetries \cite{Setia2020,Picozzi2023}. These encode symmetries as a system of Boolean linear equations in the spin-orbital occupancies, or equivalently as Pauli operators that commute exactly with the Hamiltonian, and then remove the redundant qubits via Clifford transformations \cite{Bravyi2017,Leeuwen2024}. Because these symmetries are exact, the ground-state energy is left unchanged when projecting to the appropriate eigenspace.

The same formalism can be applied to approximate symmetries, at the cost of a loss of accuracy in the eigenspectrum. In this work, we treat the CAS approximation as imposing approximate $Z$-symmetries that fix frozen-core and virtual orbital occupancies, and we refer to the resulting combined approach as SAE-CAS. We show this yields exactly the same Hamiltonian as the canonical CAS Hamiltonian with frozen-core and virtual-orbital projection in molecular electronic-structure theory \cite{Helgaker2000}; a proof of this is provided in Appendix~\ref{app:cas-equivalence}. Further, we integrate this with exact point-group and spin-parity symmetry qubit removal to maximise qubit savings.

The Jordan--Wigner mapping is one of several fermion-to-qubit encodings \cite{Steudtner2018}; the Bravyi--Kitaev (BK) mapping \cite{Bravyi2002,Seeley2012,Tranter2018} is another widely used option, designed to balance the locality of fermionic creation and annihilation operators between qubits in a way that scales as $\mathcal{O}(\log n)$ rather than $\mathcal{O}(n)$. Because BK is itself a Clifford basis change of JW that acts as an affine map on computational-basis bitstrings, it fits naturally into the same framework: appending the JW-to-BK basis change as a final affine Clifford on the active-space qubits yields an SAE-CAS-BK encoding that is unitarily equivalent to SAE-CAS and uses the same number of qubits, Pauli terms and variational parameters, while changing the per-term Pauli weight and entangling-gate counts of the resulting circuits.

This paper is structured as follows. In Section~\ref{sec:method} we present the formalism for SAE-CAS, including the second-quantised molecular Hamiltonian, the CAS projector, Boolean spin-parity and point-group symmetries, affine Clifford maps, their use for symmetry-based qubit removal and CAS projection, and the composition with the Bravyi--Kitaev mapping. In Section~\ref{sec:results} we benchmark SAE-CAS on a variety of molecules using UCCSD and a hardware-efficient shifted-circular-alternating (HE-SCA) VQE ansatz, and we compare against Jordan--Wigner (JW), JW-CAS, and JW-CAS with symmetry filtering \cite{Romero2018}; we additionally report a head-to-head comparison of SAE-CAS and SAE-CAS-BK. Finally, in Section~\ref{sec:conclusion} we conclude with an outlook on resource-efficient quantum chemistry simulations.

\section{Method}
\label{sec:method}

\subsection{Molecular electronic Hamiltonian.}
\label{subsec:hamiltonian}
The (non-relativistic, Born--Oppenheimer) molecular electronic structure Hamiltonian in second quantisation is
\begin{equation}
{H}
= \sum_{pq} h_{pq}\, {a}_p^\dagger {a}_q
+ \tfrac12 \sum_{pqrs} g_{prqs}\, {a}_p^\dagger {a}_q^\dagger {a}_s {a}_r
+ h_0 .
\label{eq:hamiltonian}
\end{equation}
where $ a^\dagger_p$ and $ a_q$ are the fermionic creation and annihilation operators, and the one- and two-electron integrals $h_{pq}$ and $g_{pqrs}$ and the nuclear--nuclear repulsion constant $h_0$ are defined as \cite{Helgaker2000}
\begin{align}
\label{eq:hamiltonian2}
h_{pq}
&= \int \phi^*_p(\mathbf{x}) 
\left( -\frac{\hbar^2}{2m}\nabla^2 - \sum_k \frac{e^2 Z_k}{4\pi\varepsilon_0  r_k} \right)
\phi_q(\mathbf{x})   d\mathbf{x},\\
g_{pqrs}
&= \frac{e^2}{4 \pi \varepsilon_0}
\iint \frac{\phi^*_p(\mathbf{x}_1) \phi_q(\mathbf{x}_1) \phi^*_r(\mathbf{x}_2) \phi_s(\mathbf{x}_2)}
{r_{12}}   d\mathbf{x}_1   d\mathbf{x}_2,\\
h_0 &= \frac{e^{2}}{4 \pi \varepsilon_{0}} \sum_{A<B} \frac{Z_{A} Z_{B}}{R_{AB}}
\end{align}

Here, $\phi_p(\mathbf{x})$ are orthonormal spin-orbitals with $\mathbf{x}=(\mathbf{r},\sigma)$ and $d\mathbf{x}$ denoting integration over space $\mathbf{r}$ and a sum over the spin degree of freedom $\sigma$, $\hbar$ is the reduced Planck constant, $m$ the electron mass, $e$ the elementary charge, $\varepsilon_0$ the vacuum permittivity, and $Z_k$ the charge of nucleus $k$.

\subsection{The complete active space (CAS) approximation.}
\label{subsec:cas}
In the complete active space (CAS) approach, the spin-orbital basis is split into three subspaces: the frozen---core orbitals which are assumed to be occupied; the active orbitals whose occupations are allowed to vary to capture correlation; and the virtual orbitals that are assumed to be unoccupied. Defining projectors ${P}_F = \bigotimes_{i\in F} \ket{1} \bra{1}_i,$ and ${P}_V = \bigotimes_{i\in V} \ket{0} \bra{0}_i$, the complete active space projector is 
\begin{equation}
     P
= {P}_F
\otimes \mathbf{1}_A \otimes 
{P}_V,
\end{equation}
\qquad
where the sets $F$, $A$ and $V$ correspond respectively to the indices of the frozen-core, active and virtual spin orbitals. The complete active space Hamiltonian is 
\begin{equation}
    {H'}= P  H  P
\end{equation}, where $H$ is the molecular electronic structure Hamiltonian in Eq. \eqref{eq:hamiltonian}.

Then, the constant term and the one-electron integrals of $H'$ can be written in terms of those of $H$ as:
\begin{align}
h_0'&=h_0+\sum_{i\in F} h_{ii}
+\tfrac12\sum_{i,j\in F}\bigl(g_{iijj}-g_{ijji}\bigr),
\label{eq:fcc-spin}\\
h_{pq}'&=h_{pq}+\sum_{i\in F}\bigl(g_{p i q i}-g_{p i i q}\bigr),
\label{eq:fc1ei-spin}
\end{align}
where $p,q\in A$. All two-electron terms with indices entirely in $A$  are unchanged.

The Hartree--Fock energy is given by
\begin{equation}
E_{\mathrm{HF}}
= h_0 + \sum_{i\in \mathrm{occ}} h_{ii}
+ \tfrac{1}{2}\sum_{i,j\in \mathrm{occ}}\!\bigl(g_{iijj}-g_{ijji}\bigr),
\label{eq:ehf-spin}
\end{equation}
where $\mathrm{occ}$ denotes the spin-orbitals occupied in the Hartree--Fock state. If the frozen-core is taken to be these occupied spin-orbitals and there are no active orbitals, the projected Hamiltonian $ H'= P H P$ reduces to the scalar $E_{\mathrm{HF}}$.

These results are proven in Appendix~\ref{app:cas-second-quantised-proof}. 

\subsection{Spin-parity symmetries, Boolean point-group and their character table}
\label{subsec:boolean-syms}
The electronic Hamiltonian \eqref{eq:hamiltonian} commutes with several symmetries that can be encoded as binary constraints on spin--orbital occupations. Unless stated otherwise, all matrix--vector products below are over $\mathbb F_2$ (arithmetic mod 2), and bitstrings $a\in\mathbb F_2^n$ denote spin--orbital occupancies in the Jordan--Wigner computational basis.

\paragraph{Spin symmetries and parities.}
In the non-relativistic setting, $H$ commutes with $S^2$ and $S_z$. For binary constraints we use their $\mathbb Z_2$ parities
\begin{align}
P_\alpha &= (-1)^{N_\alpha},\\
P_\beta  &= (-1)^{N_\beta},\\
P_N      &= (-1)^{N}=P_\alpha P_\beta ,
\end{align}
where $N_\alpha, N_\beta$ count $\alpha$ and $\beta$ electrons. Together they form a group isomorphic to the Klein four group $\mathbb{Z}_2^2$. Any two are independent; their fixed eigenvalues are determined by the chosen $N$ and $M_S$.

\paragraph{Point-group symmetry.}
A molecule’s point group is the group of symmetry operations (rotations, reflections, inversion and improper rotations, together with the identity) that leave its geometry unchanged about a fixed point. Its irreducible representations govern symmetry-adapted orbitals, normal-mode degeneracies, and spectroscopic selection rules. 

A Boolean group is one in which every element squares to the identity; the eight Boolean molecular point groups are $C_1$ (the trivial group), $C_s$, $C_2$ and $C_i$ (each isomorphic to $\mathbb{Z}_2$), $C_{2v}$, $C_{2h}$ and $D_2$ (each isomorphic to $\mathbb{Z}_2^2$) and $D_{2h}$ (isomorphic to $\mathbb{Z}_2^3$). In practice, in quantum chemistry we descend to the largest Boolean subgroup of the full point group to make the irreducible representations one-dimensional and real: in this case each symmetry-adapted orbital is either symmetric (the symmetry acts on the orbital as multiplication by $+1$) or antisymmetric ($-1$) under each symmetry. The character tables of each point group summarise this behaviour. Under the Jordan-Wigner mapping, each symmetry is mapped to the product of Pauli $Z$ operators on each qubit that corresponds to a spin-orbital that is antisymmetric with respect to it. An example for the water molecule is provided in Appendix \ref{app:water_example}.

\paragraph{Orbital symmetry matrix.}
Let the full Boolean symmetry group be the one generated by the Boolean point-group and spin-parity symmetries. Define the orbital symmetry matrix
\begin{equation}
S_{ij}=\begin{cases}
1,& \text{spin-orbital } j \text{ has character } -1 \text{ under } g_i,\\
0,& \text{otherwise}.
\end{cases}
\label{eq:orbital-character-matrix}
\end{equation}
Then the many-electron eigenvalue of the symmetry $g_i$ on a Slater determinant with occupations $a\in\mathbb F_2^n$ is
\begin{equation}
\lambda_i(a)=(-1)^{(Sa)_i}.
\end{equation}
$P_{\alpha}$ and $P_{\beta}$ each contribute a row to $S$: the $\alpha$ row has ones on $\alpha$ spin-orbitals and the $\beta$ row has ones on $\beta$ spin-orbitals.

Since each generator $g_i$ maps under Jordan--Wigner to a Pauli $Z$-string that commutes exactly with $H$, the set $\{g_1,\dots,g_k\}$ constitutes the stabiliser group of an $[[n,n{-}k]]$ stabiliser code whose code space is the target symmetry sector. This observation suggests an alternative use of the generators: retaining all $n$ qubits and employing them for \emph{symmetry verification}~\cite{Bonet-Monroig2018,McArdle2019err} --- measuring each $g_i$ after circuit execution and post-selecting on shots that return the correct eigenvalues, thereby discarding shots corrupted by errors that anticommute with at least one generator. However, this use of the generators is less advantageous than qubit removal. Because all stabilisers are $Z$-type, they are entirely insensitive to $Z$-type (dephasing) errors, which yield a trivial syndrome and go undetected; only errors with $X$ or $Y$ support on a qubit appearing in the corresponding generator are detectable. More fundamentally, a generator used for qubit removal permanently eliminates one qubit from the circuit, reducing circuit depth, entangling-gate count, and Hamiltonian Pauli weight at the source of the noise and without statistical overhead from post-selection, whereas a generator reserved for parity verification leaves that qubit and its noise in the circuit and incurs a shot overhead from post-selection. Using all $k$ independent generators for qubit removal, as we do throughout, therefore maximises these resource savings.

\subsection{Affine Clifford operators and their tableaux}
A Clifford operator $C$ is a unitary that, under conjugation, maps Pauli operators to Pauli operators (up to phases $\pm 1,\pm i$) and preserves their multiplication structure.

Given the Clifford tableau for $C$
\begin{equation}
M=\begin{pmatrix}M_{ZZ}&M_{ZX}\\ M_{XZ}&M_{XX}\end{pmatrix},\qquad
s=\begin{pmatrix}s_Z\\ s_X\end{pmatrix},
\label{eq:general_clifford1}
\end{equation}
one can read off the action of $C$ on the single-qubit Pauli generators $Z_i$ and $X_i$ from the $i$-th and $(i+n)$-th columns of $M$, with the corresponding entries of $s$ fixing the phases. Because Cliffords preserve Pauli multiplication, this completely determines their action on all Pauli strings.

We previously proved that any qubit encoding whose action on Jordan--Wigner computational-basis bitstrings is the affine map $a \mapsto q = T a \oplus b$ over $\mathbb{F}_2$ is implemented by the Clifford $C$ with tableau
\begin{equation}
M=\begin{pmatrix}(T^{-1})^{\mathsf T}&0\\ 0&T\end{pmatrix},\qquad
s=\begin{pmatrix}T^{-1} b\\ 0\end{pmatrix},
\label{eq:affine-clifford-tableau}
\end{equation}
so $Z$’s map to products of $Z$’s (with phases set by $s$) and $X$’s to products of $X$’s. These affine Cliffords are precisely those that permute computational-basis states; equivalently, they are exactly those generated by CNOTs and $X$ gates. The definition of a Clifford tableau and compact proofs of these result are provided in Sec.~\ref{app:clifford-tableau}. Further properties of affine Clifford maps are collected in Sec.~\ref{app:clifford-maps}.

\subsection{Symmetry reduction via an affine Clifford map.}
\label{subsec:symmetry-clifford}

We encode the occupations of $n$ spin-orbitals into $n-k$ qubits by exploiting $k$ independent Boolean symmetries. 

\begin{enumerate}[label=(\roman*)]

\item \textit{Choose the generators.}
Choose $g_1,\dots,g_k$ be $k$ independent generators for the Boolean symmetry group, so that (after a reordering of qubits) $g_i$ has eigenvalue $-1$ on qubit $i$, while $g_j$ has eigenvalue $+1$ there for all $j\neq i$. Let $S\in\{0,1\}^{k\times n}$ be the binary orbital-character matrix $S$ as in Eq.~\eqref{eq:orbital-character-matrix} for the generators $g_i$. Fix the target sector (the simultaneous eigenspace of the Boolean symmetry group of interest) by choosing $c\in\mathbb F_2^k$ so that occupation in the correct target $a\in\mathbb F_2^n$ satisfy
\begin{equation}
Sa \;=\; c .
\label{eq:sector-constraint}
\end{equation}
In practice, the canonical set of generators is found through row-reduction of the orbital-character matrix of a (possibly larger) set of generators.

\item \textit{Define the affine map.}
Let $T\in\{0,1\}^{n\times n}$ be the identity with its first $k$ rows replaced by those of $S$, and set
\begin{equation}
b \;=\; \binom{c}{0_{n-k}} .
\label{eq:b-definition}
\end{equation}
By construction $T$ has full rank and is therefore invertible over $\mathbb F_2$ (i.e.\ $T\in GL(n,2)$).

\item \textit{Change basis with the affine Clifford.}
There is a Clifford $C$ that realises the affine action
\begin{equation}
\ket{a}\ \mapsto\ \ket{q} \;=\; C\ket{a} \;=\; \ket{T a \oplus b},
\label{eq:affine-action}
\end{equation}
and whose tableau is as in Eq.~\eqref{eq:affine-clifford-tableau}; this is proved in Appendix~\ref{app:sae-clifford}. Under $C$, each generator $g_i$ is mapped (by conjugation) to a single-qubit $Z_i$ on the first $k$ qubits (up to a phase fixed by the tableau), and every state in the sector \eqref{eq:sector-constraint} is mapped to one with the first $k$ bits equal to $0$.

\item \textit{Project and drop symmetry qubits.}
Conjugate the Hamiltonian and project onto the chosen sector:
\begin{equation}
H' \;=\; P\, C\, H\, C^\dagger\, P .
\label{eq:clifford-projection}
\end{equation}
In practice we apply the tableau of $C$ to each Pauli term of $H$; discard any term that has $X$ or $Y$ on the first $k$ qubits; replace $Z_i$ by $1$ for $1\le i\le k$; then remove those $k$ qubits. The resulting $(n-k)$-qubit Hamiltonian is isospectral to $H$ within the selected symmetry sector.

\end{enumerate}

\subsection{SAE-CAS via approximate $Z$-symmetries.}
In SAE-CAS, the projection onto frozen-core and virtual orbitals is implemented as approximate $Z$-symmetries on the corresponding qubits. One first applies the exact-symmetry symmetry-adapted encoding to reduce qubits, and then projects out the frozen or virtual qubits to obtain the active-space qubit Hamiltonian.

The advantage of this approach is that it composes naturally with symmetry-based qubit removal.
In practice, this is done as in Eq.~\eqref{eq:clifford-projection}, where the overall Clifford for the encoding is $C=C_2 C_1$. Here $C_1$ is the Clifford constructed in \ref{subsec:symmetry-clifford}, with the restriction that the symmetry bits are only allowed to correspond to active space spin-orbitals (when we row-reduce $S$ we pivot on the active-space bits and discard the remaining rows). The operator $C_2$ is the Clifford implementing the complete-active-space approximation: its tableau has $T_2=I$ and $b_2$ with nonzero entries exactly on the frozen-core spin-orbitals (so that each approximate $Z$-symmetry acts on a single spin-orbital, taking eigenvalue $-1$ for frozen-core qubits and $+1$ for virtual qubits). We prove in Appendix~\ref{app:compatibility} that exact-symmetry affine basis changes commute with CAS qubit removal, so the order of applying $C_1$ and the CAS projection does not affect the resulting active-space Hamiltonian.

In Appendix~\ref{app:clifford-composition} we prove the general composition law for affine Clifford maps. In the present setting (with $T_2=I$), this yields
\begin{equation}
M=\begin{pmatrix} (T_1^{-1})^{T}&0\\0&T_1\end{pmatrix},
\qquad
s=\begin{pmatrix}T_1^{-1} (b_1 \oplus b_2)\\0\end{pmatrix}.
\end{equation}

Unlike the case of qubit removal using only exact symmetries, introducing CAS approximate $Z$-symmetries induces an approximation error in the energies. However, this error is identical to that of the canonical complete-active-space approximation. Indeed, the Hamiltonian obtained by (i) mapping the full problem to qubits via the Jordan--Wigner transform and then removing the qubits corresponding to frozen-core and virtual orbitals is exactly the same as (ii) forming the CAS Hamiltonian of \ref{subsec:cas} and then applying the Jordan--Wigner mapping on the active orbitals only. This equivalence is shown explicitly in Appendix~\ref{app:cas-equivalence} and illustrated on several examples in Section~\ref{sec:results}.

\subsection{Combination with the Bravyi--Kitaev mapping.}
\label{subsec:sae-cas-bk}

The Bravyi--Kitaev (BK) mapping \cite{Bravyi2002,Seeley2012} is an alternative fermion-to-qubit encoding to Jordan--Wigner. Whereas in JW each qubit stores the occupation number of a single spin-orbital, in BK each qubit stores instead the parity of the occupation numbers of a particular subset of orbitals; in this way both occupation-number and parity information are spread across the register, so that the Pauli weight of fermionic creation and annihilation operators scales as $\mathcal{O}(\log n)$ rather than $\mathcal{O}(n)$. Concretely, the JW-to-BK basis change is a Clifford operator $C_{\mathrm{BK}}$ whose action on a Jordan--Wigner computational-basis bitstring $f\in\mathbb F_2^n$ (the occupation numbers of the $n$ spin-orbitals) is the affine map $f\mapsto T_n f$, where the binary BK matrix $T_n$ is defined recursively, on registers of size $n=2^x$, by $T_1=(1)$ and
\begin{equation}
T_{2n}=
\begin{pmatrix}
T_n & 0\\
A_n & T_n
\end{pmatrix},
\label{eq:bk-matrix-recursion}
\end{equation}
with $A_n$ the $n\times n$ binary matrix whose last row is all ones and whose other entries are zero. Reading off the rows of $T_n$, each BK qubit thus encodes the parity of a particular subset of orbitals: most qubits store the occupancy of a single orbital, while at each recursion level one ``summary'' qubit (the last of each block) stores the total parity of all preceding orbitals in that block. $C_{\mathrm{BK}}$ satisfies the form of Eq.~\eqref{eq:affine-clifford-tableau} with $T=T_n$ and $b=0$, and its tableau is
\begin{equation}
M_{\mathrm{BK}}=\begin{pmatrix}(T_n^{-1})^{\mathsf T}&0\\ 0&T_n\end{pmatrix},
\qquad s_{\mathrm{BK}}=0 .
\label{eq:bk-clifford-tableau}
\end{equation}
The invertibility of $T_n$ over $\mathbb F_2$ and the derivation of \eqref{eq:bk-clifford-tableau} from Eq.~\eqref{eq:affine-clifford-tableau} are recorded in Appendix~\ref{app:bk-affine-clifford}.
The SAE-CAS-BK encoding is obtained by appending $C_{\mathrm{BK}}$ to the SAE-CAS Clifford on the reduced active subspace, $C\to C_{\mathrm{BK}}\, C_2\, C_1$, and using the composition law of Eq.~\eqref{eq:clifford-theorem} to assemble a single affine-Clifford tableau. Because $C_{\mathrm{BK}}$ is unitary and acts only on the $n-k$ active-space qubits, SAE-CAS-BK is unitarily equivalent to SAE-CAS: it has the same number of qubits, the same number of distinct Pauli terms in the qubit Hamiltonian, the same number of UCCSD variational parameters, and the same eigenspectrum within the chosen sector. The two encodings differ only in the locality structure of the qubit Hamiltonian (the per-term Pauli weight) and consequently in circuit depth and CNOT count. The motivation for SAE-CAS-BK is therefore the standard motivation for BK over JW: at sufficiently large active spaces the logarithmic per-term weight is expected to give shorter Trotterised excitation circuits and lower entangling-gate counts than JW-based SAE-CAS, while inheriting all qubit and Pauli-term reductions of the SAE-CAS construction. More generally, this construction illustrates that any fermion-to-qubit encoding that can be expressed as an affine Clifford basis change of Jordan--Wigner---of which Bravyi--Kitaev is one example, alongside parity and other variants \cite{Seeley2012}---can be composed with SAE-CAS qubit removal in exactly the same way, by appending its Clifford tableau to $C_2 C_1$ on the active subspace.

\section{Numerical results}
\label{sec:results}

To evaluate the SAE-CAS approach across a representative set of small molecules, we consider water (\ce{H2O}), ethene (\ce{C2H4}), oxygen (\ce{O2}), methylene (\ce{CH2}), carbon monoxide (\ce{CO}), formaldehyde (\ce{H2CO}), nitrogen (\ce{N2}), cyclobutadiene (\ce{C4H4}) and 1,3-butadiene (\ce{C4H6}). This selection spans a range of point-group symmetries, bonding configurations and active-space sizes that are representative of small-molecule electronic-structure problems. All geometries are taken at their experimental equilibrium configurations in a minimal STO-3G basis. For each molecule the active space is taken as the top-occupation MP2 natural orbitals for closed-shell strongly correlated systems (\ce{H2O}, \ce{C2H4}, \ce{H2CO}, \ce{C4H4}, \ce{C4H6}) and the HOMO/LUMO frontier RHF window for the remaining open-shell or diatomic systems (\ce{O2}, \ce{CH2}, \ce{CO}, \ce{N2}); the corresponding $(n,m)$-CAS sizes are listed in Tables~\ref{tab:resources-part1}--\ref{tab:resources-part2}.

The variational algorithms employ two distinct ansätze. First, the unitary coupled-cluster with single and double excitations (UCCSD) ansatz constructs the trial state by exponentiating fermionic excitation operators corresponding to single and double excitations, yielding a chemically motivated circuit \cite{Bartlett2007,Peruzzo2014,Romero2018}. Second, the hardware-efficient shifted-circular-alternating (HE-SCA) ansatz, consisting of repeated layers of parameterised $R_Y$ rotations (one on each qubit) followed by CNOT entangling blocks; we use the shifted-circular-alternating pattern characterised in \cite{Kandala2017,HE-SCA19}. For HE-SCA we scanned the number of layers and, for each molecule, retained the best-performing circuit (lowest converged energy): for the four-orbital active spaces (\ce{H2O}, \ce{C2H4}, \ce{CH2}, \ce{H2CO}, \ce{N2}, \ce{C4H4}, \ce{C4H6}) we tested $0$--$60$ layers; for the five-orbital active spaces (\ce{O2}, \ce{CO}) we tested $0$--$200$ layers.

For UCCSD we compare four strategies:
\begin{enumerate}
  \item JW: the standard Jordan--Wigner mapping without any symmetry-based qubit removal.
  \item JW-CAS: the Jordan--Wigner mapping of the complete active-space orbitals only, with excitations restricted to the chosen active orbitals but not exploiting any point-group symmetries.
  \item JW-CAS (SF): the same active-space encoding as JW-CAS with manual screening of excitations to enforce point-group symmetry by including only operators in the totally symmetric representation; this matches the number of variational parameters of SAE-CAS and reduces circuit complexity.
  \item SAE-CAS: the symmetry-adapted qubit encoding with complete active space mapping which incorporates both point-group and spin-parity symmetries to remove redundant qubits.
\end{enumerate}
For the HE-SCA ansatz we compare Jordan--Wigner mapping of the complete active-space orbitals only (JW-CAS) and our proposal, SAE-CAS.

Tables~\ref{tab:resources-part1} and \ref{tab:resources-part2} report, for each encoding, (a) number of qubits and (b) Hamiltonian Pauli count (also shown in Fig.~\ref{fig:resources_a}\subref{subfig:qubits},~\subref{subfig:pauli_weight}) together with (c) circuit depth, (d) CNOT count, (e) number of variational parameters, (f) VQE iterations, (g) VQE ground-state energy, and (h) the CAS(\textsf{PySCF}) reference (where circuit resource metrics are shown in Fig.~\ref{fig:resources_b} and convergence and accuracy are shown in Fig.~\ref{fig:performance}). All classical optimisations were performed with the SLSQP algorithm \cite{Kraft1988}.

Across molecules and both ansätze, incorporating symmetry at the {encoding} level provides consistent and substantial savings. Relative to JW-CAS, SAE-CAS removes additional qubits, lowers the Pauli count of the Hamiltonian, and yields shallower, less entangling circuits with fewer parameters. Enforcing point-group symmetry only at the circuit level (JW-CAS (SF)) helps reduce circuit depth and parameter count but leaves the qubit register and Pauli count unchanged from JW-CAS; SAE-CAS still outperforms it across all resource metrics, because the symmetry sector is enforced natively in the qubit map rather than post hoc in the ansatz.

For UCCSD with $\theta=0$ initialisation, the three encodings (JW-CAS, JW-CAS (SF) and SAE-CAS) all reach essentially the same converged energy on every molecule for which all three could be evaluated, with residual differences at the level of $10^{-8}\,E_h$ in nearly all cases. The exception is the JW-CAS (SF) ansatz for \ce{N2}, where the Trotterised circuit on the full eight-qubit JW-CAS register reaches a minimum $\sim 2\times10^{-3}\,E_h$ above SAE-CAS: the same fermionic excitations, applied without symmetry-based qubit removal, generate a larger set of noncommuting qubit operators on the unreduced register whose Trotter approximation no longer reproduces the symmetric-sector minimum (\ce{CO} shows a much milder version of the same effect). What distinguishes SAE-CAS at the UCCSD level is therefore not the absolute energy, but the cost at which that energy is reached: SAE-CAS uses 2--3$\times$ fewer SLSQP iterations than JW-CAS (e.g.\ 88 vs 278 for \ce{CO}, 21 vs 109 for \ce{N2}, 111 vs 195 for \ce{H2O}), in addition to the reduced qubit, Pauli, depth, CNOT and parameter counts. For \ce{C4H4} and \ce{C4H6} we report SAE-CAS numbers but were unable to evaluate JW-CAS or JW-CAS (SF) UCCSD because constructing the corresponding ansatz exhausts memory (footnote a in Table~\ref{tab:resources-part2}); SAE-CAS does not have this issue because symmetry-based qubit removal filters the excitation list to a much smaller subset at the mapping stage.

For HE-SCA, we find a particularly clear separation: with SAE-CAS, every molecule in our benchmark set converges within the layer budgets above to the CAS reference with errors in the range $\mathcal{O}(10^{-6})\,E_h$ (see the ``VQE energy error'' rows). With JW-CAS the four-orbital systems also converge, but require five to six times more layers than SAE-CAS to reach chemical accuracy (e.g.\ \ce{C4H6} reaches $\sim 10^{-3}\,E_h$ at 30 layers under JW-CAS versus 6 layers under SAE-CAS). The five-orbital cases are qualitatively different: for \ce{O2} (6,5) and \ce{CO} (6,5), the JW-CAS HE-SCA simulations did not converge within the 0--200 layer budget; the corresponding entries are left blank in Table~\ref{tab:resources-part1}. Two diagnostics in Appendix~\ref{app:hesca-limitations} identify the mechanism. The gradient variance $\mathrm{Var}_\theta[\partial\langle H\rangle/\partial\theta_0]$ on the 10-qubit JW-CAS register drops two orders of magnitude from $L{=}1$ to $L{=}5$ and then saturates at $\sim 10^{-3}$ across $L\in[25,200]$, ruling out a classical barren-plateau (no exponential decay with depth). Repeated SLSQP runs from random initialisations on the JW-CAS HE-SCA ansatz then reveal a non-particle-number-preserving drift: HE-SCA (a circuit of $R_y$ and CNOT gates) does not commute with $(N_\alpha, N_\beta)$, and without symmetry-based qubit removal at the encoding level the variational manifold spans all spin-orbital occupation sectors. Symmetry-preserving state-preparation circuits address related leakage at the ansatz level \cite{Gard2020}; here the restriction is instead imposed by the encoding. For \ce{O2} this leads the optimiser into a wrong $(N_\alpha,N_\beta)$ sector in 6/10 trials (including one converged state with non-integer $\langle N_\alpha\rangle$, $\langle N_\beta\rangle$, i.e.\ a literal sector superposition); for \ce{CO} the sector is mostly preserved but the optimiser gets trapped in higher-energy local minima inside the target sector (mean error $+0.097\,E_h$). SAE-CAS removes the wrong-sector channel by removing spin-parity qubits in the encoding, restricting the variational manifold to a 5- or 6-qubit target symmetry sector with markedly fewer local minima.

HE-SCA with SAE-CAS provides a significantly more compact circuit than the UCCSD ansatz and, in noiseless simulations, recovers the CAS-CI ground state once a modest number of layers is used. This compactness comes with a larger raw parameter count than UCCSD (some parameters are redundant), which is known to affect convergence under noise; additionally, the optimal number of HE-SCA layers must be found empirically. Even so, the layer scan described above reliably identified convergent circuits for all systems under SAE-CAS, whereas JW-CAS lacked convergent HE-SCA solutions for \ce{O2} and \ce{CO} within the same budgets.

\begin{table}
  \centering
  \caption{Resource comparison for H\textsubscript{2}O, C\textsubscript{2}H\textsubscript{4}, CH\textsubscript{2}, O\textsubscript{2} and CO}
  \label{tab:resources-part1}
  \includestandalone[mode=buildmissing, height=0.95\textheight]{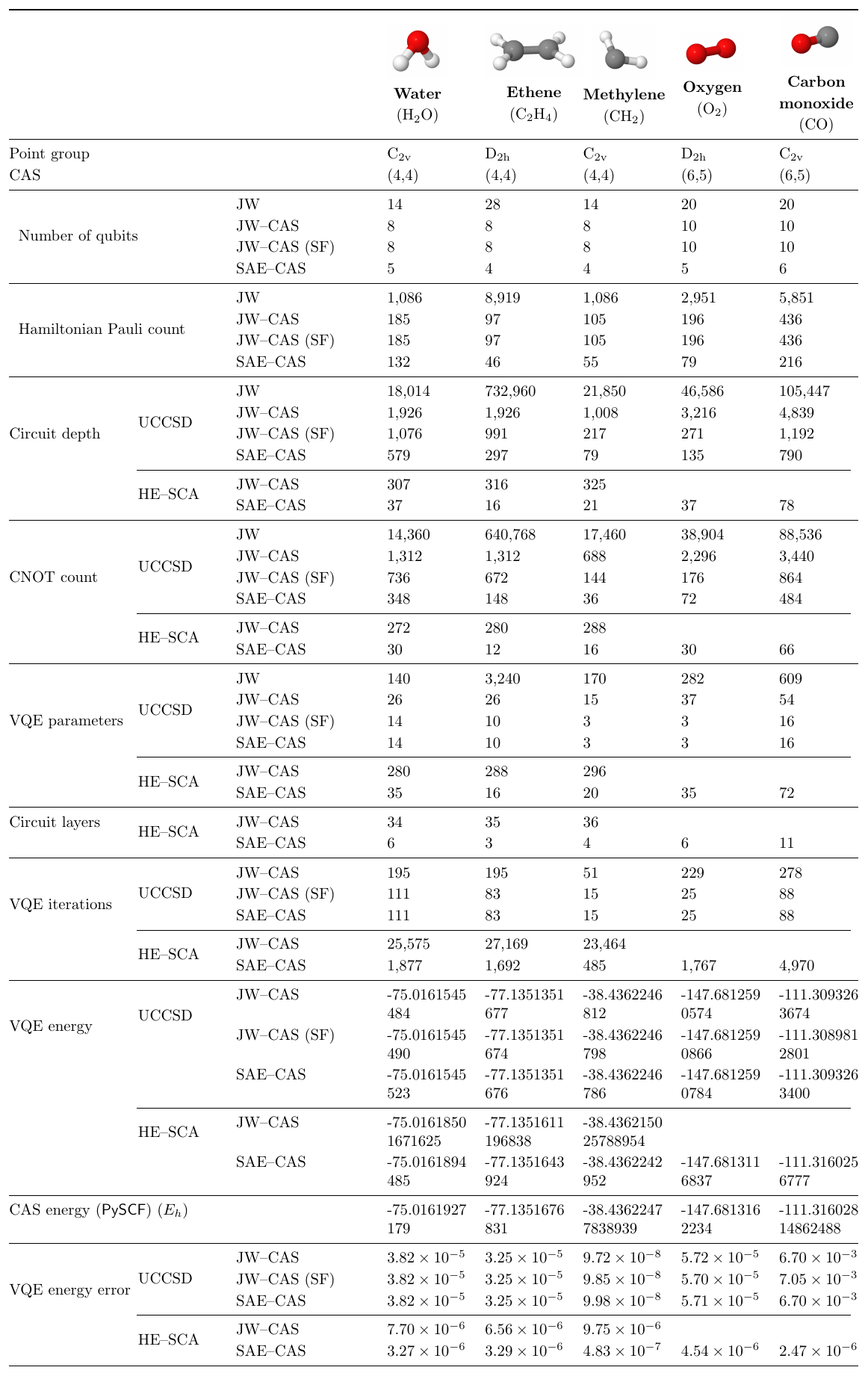}
\end{table}

\begin{table}
  \centering
  \caption{Resource comparison for H\textsubscript{2}CO, N\textsubscript{2}, C\textsubscript{4}H\textsubscript{4}, and C\textsubscript{4}H\textsubscript{6}}
  \label{tab:resources-part2}
  \includestandalone[mode=buildmissing, height=0.95\textheight]{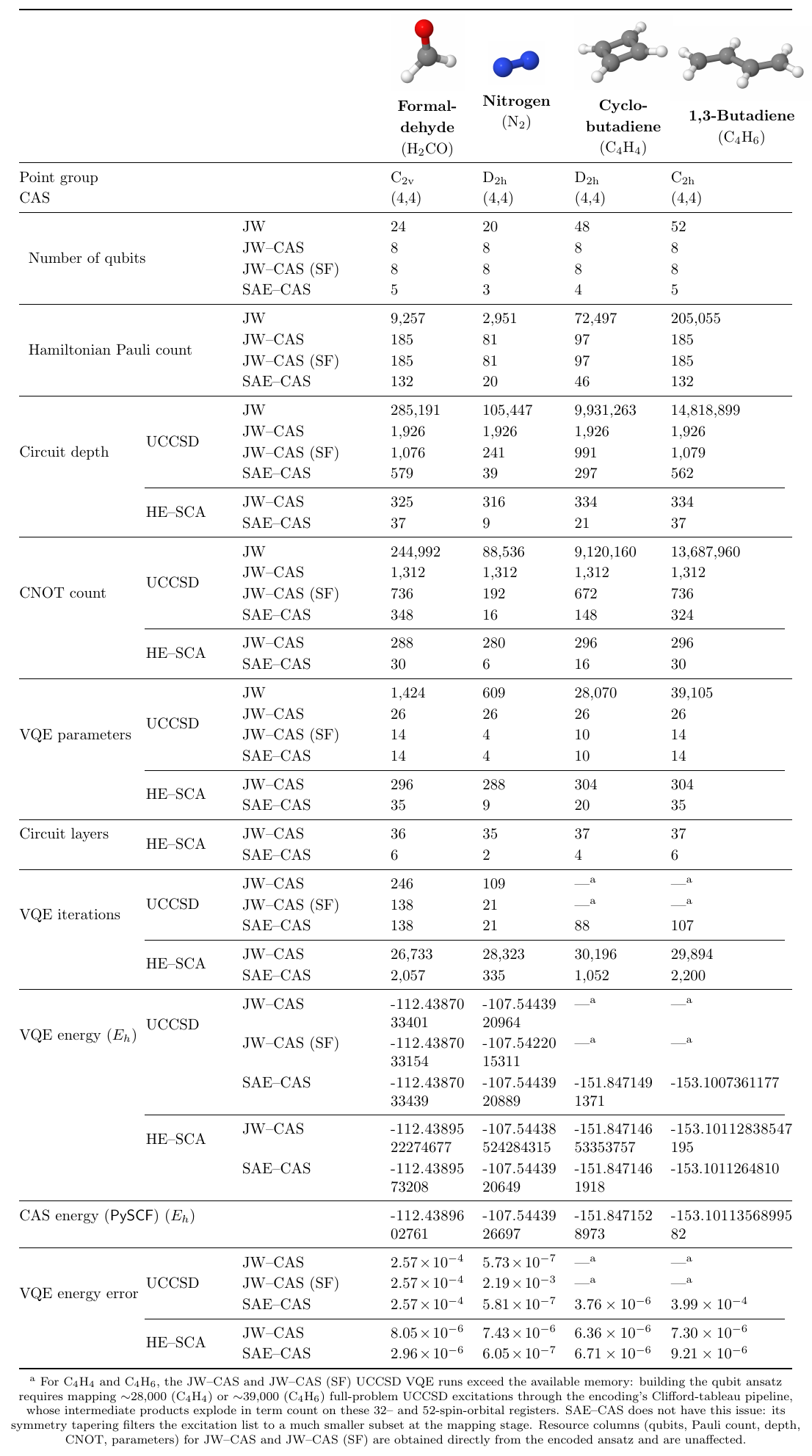}
\end{table}

\subsection{Bravyi--Kitaev variant: SAE-CAS-BK}
\label{subsec:results-bk}

We now compare SAE-CAS to its Bravyi--Kitaev composition, SAE-CAS-BK (Section~\ref{subsec:sae-cas-bk}), on the same nine molecules using the UCCSD ansatz initialised at the Hartree--Fock reference ($\theta=0$). Because the additional Bravyi--Kitaev Clifford acts only as a basis change on the reduced active subspace, every quantity invariant under such a basis change should agree between the two encodings; Table~\ref{tab:resources-bk} confirms this empirically. SAE-CAS and SAE-CAS-BK match exactly on the qubit count, Pauli term count and variational-parameter count, and produce the same converged VQE energy (to within $\sim 10^{-8}\,E_h$) in the same number of SLSQP iterations for every molecule in the benchmark set. The two encodings differ only in the basis-dependent resources, namely the per-term Pauli weight and the UCCSD circuit depth and CNOT count: SAE-CAS-BK matches SAE-CAS in CNOT count for \ce{CH2} and \ce{N2}, differs by within $\pm 6\,\%$ for the remaining molecules except \ce{CO}, and shows a $+20\,\%$ CNOT increase for \ce{CO}. On these CAS($n$,$m$) sizes the asymptotic $\mathcal{O}(\log n)$ locality advantage of BK is therefore not yet realised; we expect the choice between the two encodings to matter more at larger active spaces.

\begin{table*}
  \centering
  \caption{Comparison of SAE-CAS and SAE-CAS-BK for the same nine molecules and active spaces as Tables~\ref{tab:resources-part1}--\ref{tab:resources-part2}, with the UCCSD ansatz initialised at the Hartree--Fock reference ($\theta=0$). The two encodings are unitarily equivalent and consequently agree in qubit count, Hamiltonian Pauli count, variational-parameter count and converged VQE energy (the last to $\sim 10^{-8}\,E_h$) by construction; the only resource-level differences are confined to UCCSD circuit depth and CNOT count.}
  \label{tab:resources-bk}
  \includestandalone[mode=buildmissing, width=\textwidth]{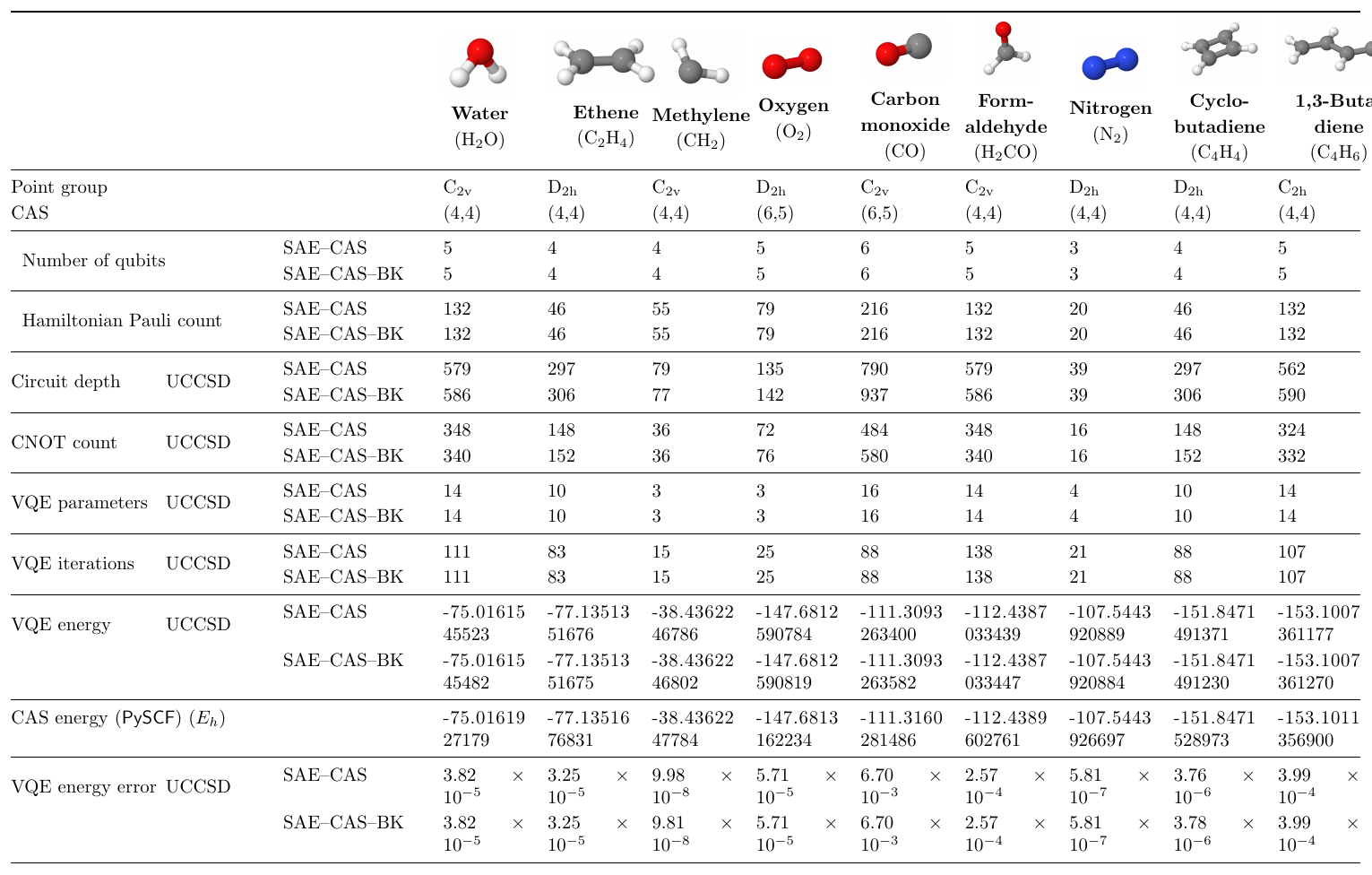}
\end{table*}

\section{Conclusion}
\label{sec:conclusion}
We introduced a symmetry-adapted qubit encoding with complete active space (SAE-CAS) that treats the CAS approximation as a set of approximate $Z$-symmetries and composes it with \emph{exact} point-group and spin-parity symmetries via affine Clifford maps. We proved that the resulting active-space qubit Hamiltonian is exactly equivalent to the canonical CAS Hamiltonian with frozen-core and virtual-orbital projection and that exact-symmetry basis changes commute with CAS qubit removal. This places SAE-CAS on firm theoretical footing: any approximation error arises solely from the chosen active space, not from the encoding itself.

Our numerical results across nine small molecules (Tables~\ref{tab:resources-part1}--\ref{tab:resources-part2}) demonstrate consistent resource savings relative to Jordan--Wigner baselines. For both UCCSD and HE-SCA ansätze, SAE-CAS reduces the number of qubits, lowers Pauli count, and yields shallower circuits with fewer entangling gates and variational parameters than JW-CAS. Enforcing point-group symmetry only at the circuit level (JW-CAS with symmetry filtering) helps, but encoding the symmetry natively (SAE-CAS) provides superior reductions across all metrics and removes redundant search directions from the variational landscape.

For the hardware-efficient HE-SCA ansatz, SAE-CAS converged to CAS energies for every molecule in our benchmark set within our layer budgets in a fraction of the layers needed by JW-CAS, and JW-CAS did not converge within the tested layer and optimisation budgets for \ce{O2} (6,5) and \ce{CO} (6,5). This highlights a practical advantage of enforcing symmetry at the encoding stage rather than post hoc in the ansatz.

We further showed that the same affine-Clifford framework composes SAE-CAS with the Bravyi--Kitaev mapping, yielding the unitarily equivalent SAE-CAS-BK encoding. Across our nine benchmarks the two encodings produce the same converged VQE energies (to $\sim 10^{-8}\,E_h$) in the same number of SLSQP iterations, as expected from unitary equivalence; they differ only in the basis-dependent UCCSD circuit depth and CNOT count, by within $\pm 6\,\%$ except for \ce{CO} ($+20\,\%$). We expect the choice between SAE-CAS and SAE-CAS-BK to become more consequential at larger active spaces, where BK's $\mathcal{O}(\log n)$ Pauli-weight scaling is known to be advantageous.

In summary, SAE-CAS offers a composable and systematic pathway to resource-efficient molecular simulations on fault-tolerant and near-term quantum processors, reducing qubits and operator complexity while preserving target symmetry sectors and CAS accuracy. We expect these savings to translate to larger active spaces and more expressive bases, especially when combined with measurement-reduction, error-mitigation, and compact ansätze. An open-source implementation of both SAE-CAS and SAE-CAS-BK is available in \texttt{QuantumSymmetry}.

\begin{acknowledgments}
D.P.'s research is supported by an EPSRC Postdoctoral Prize Fellowship [EP/W524335/1] at UCL and an EPSRC Research Fellowship [EP/S021582/1] at the London Centre for Nanotechnology. D.P.'s PhD was supported by an EPSRC Industrial CASE studentship [EP/T517793/1].
\end{acknowledgments}

\section*{Competing interests}

The Authors declare no competing financial or non-financial interests.

\section*{Data availability}

The code implementation described in the paper is available as an open-source Python package, whose source code is hosted on GitHub and archived on Zenodo (concept DOI \href{https://doi.org/10.5281/zenodo.7724696}{10.5281/zenodo.7724696})~\cite{picozzi2023quantumsymmetry}.

\section*{Author contributions}

D.P. conceived the main ideas introduced in the paper, wrote the initial version of the manuscript, and carried out the code implementation. J.T. contributed to discussions and the review of the manuscript.

\pagebreak

\begin{figure*}
  \centering

  \subfloat[UCCSD circuit depth\label{subfig:uccsd_depth}]{
    \includegraphics[width=0.48\textwidth]{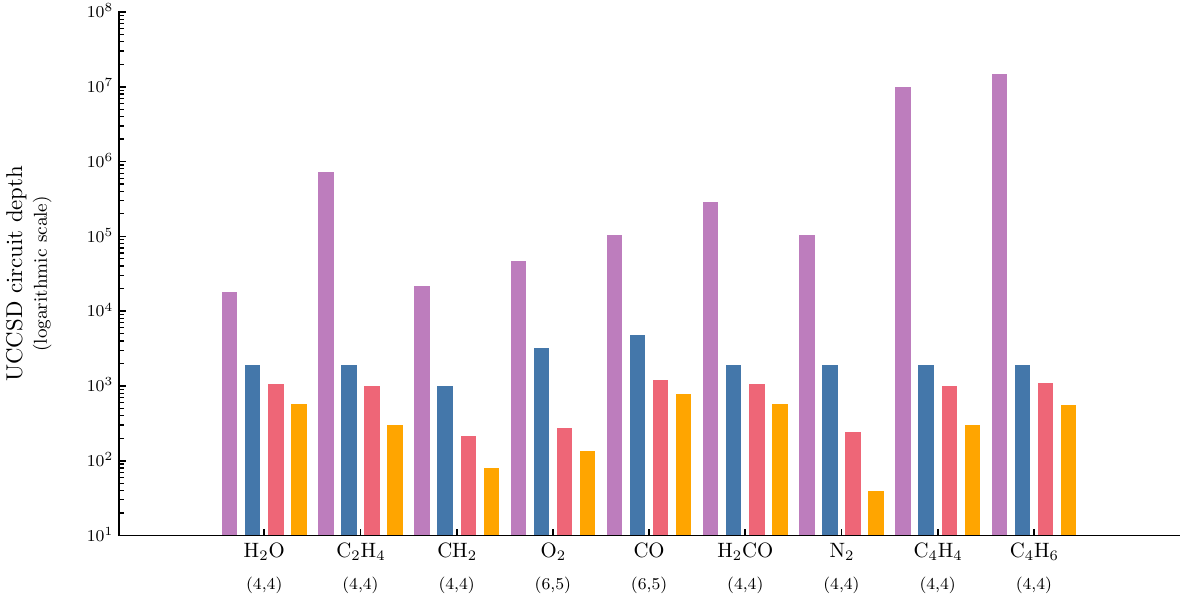}
  }\hfill
  \subfloat[HE-SCA circuit depth\label{subfig:hea_depth}]{
    \includegraphics[width=0.48\textwidth]{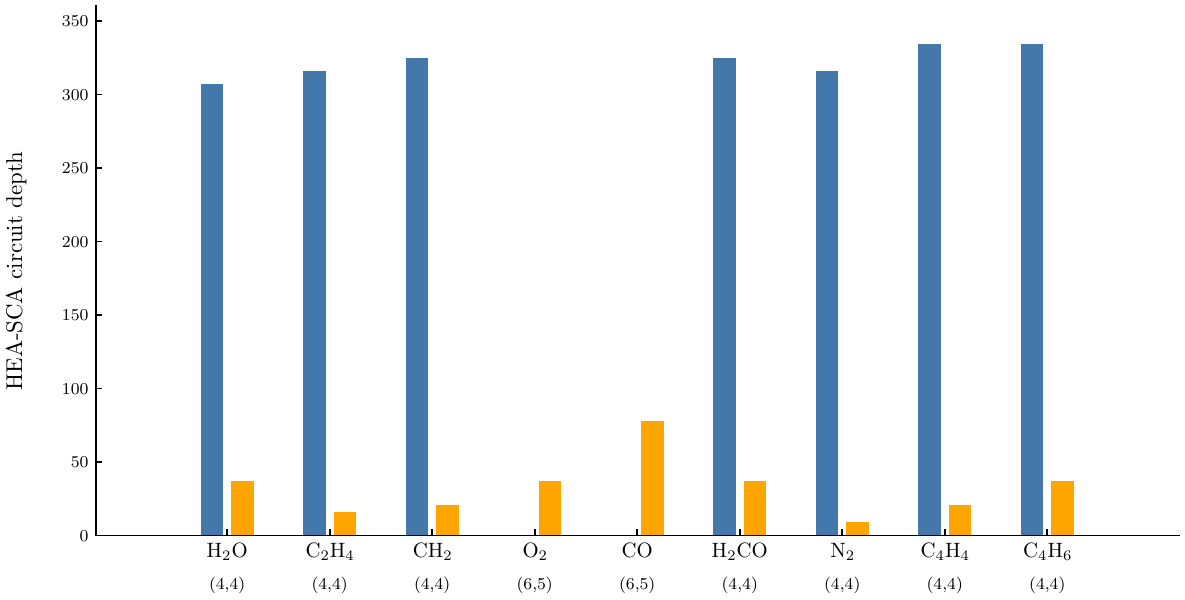}
  }\\[6pt]

  \subfloat[UCCSD CNOT count\label{subfig:uccsd_cnot}]{
    \includegraphics[width=0.48\textwidth]{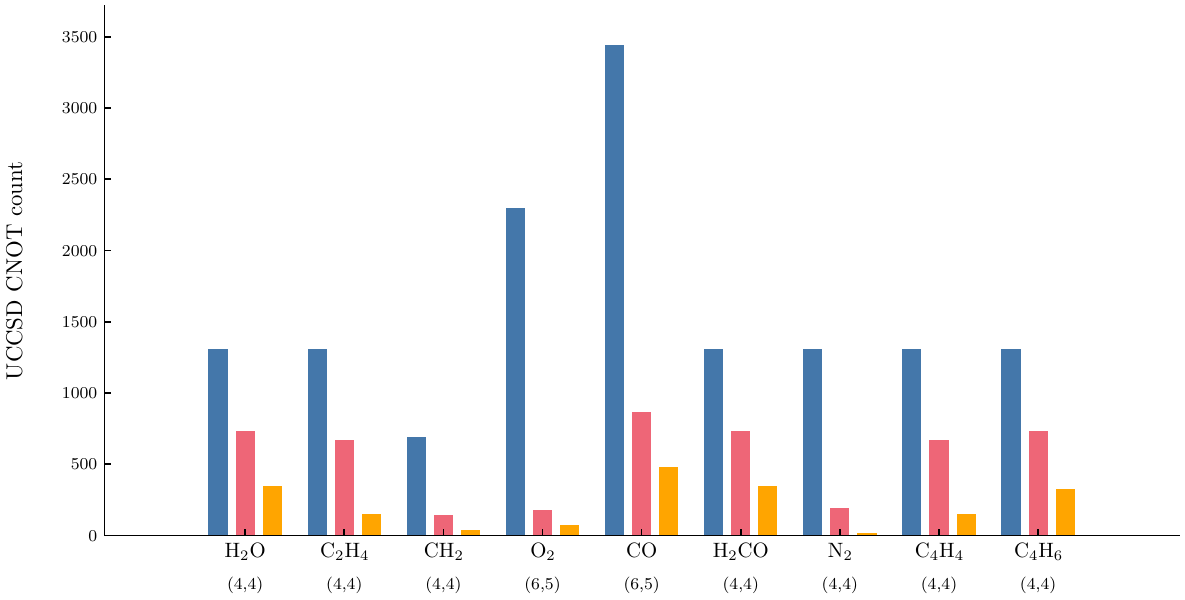}
  }\hfill
  \subfloat[HE-SCA CNOT count\label{subfig:hea_cnot}]{
    \includegraphics[width=0.48\textwidth]{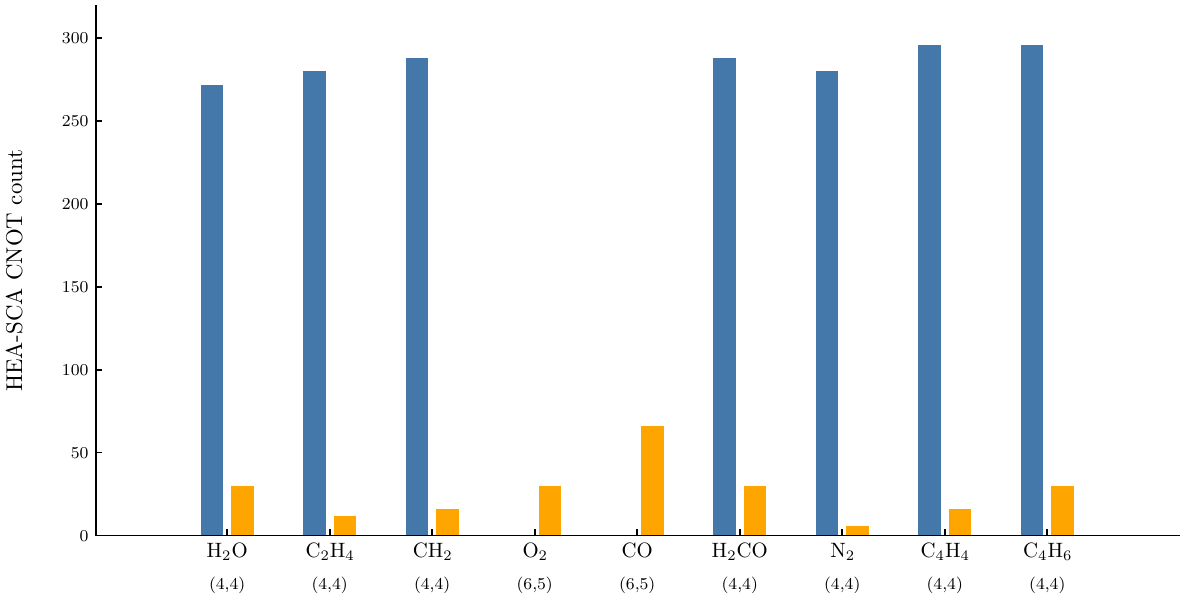}
  }\\[6pt]

  \subfloat[UCCSD VQE parameters\label{subfig:uccsd_params}]{
    \includegraphics[width=0.48\textwidth]{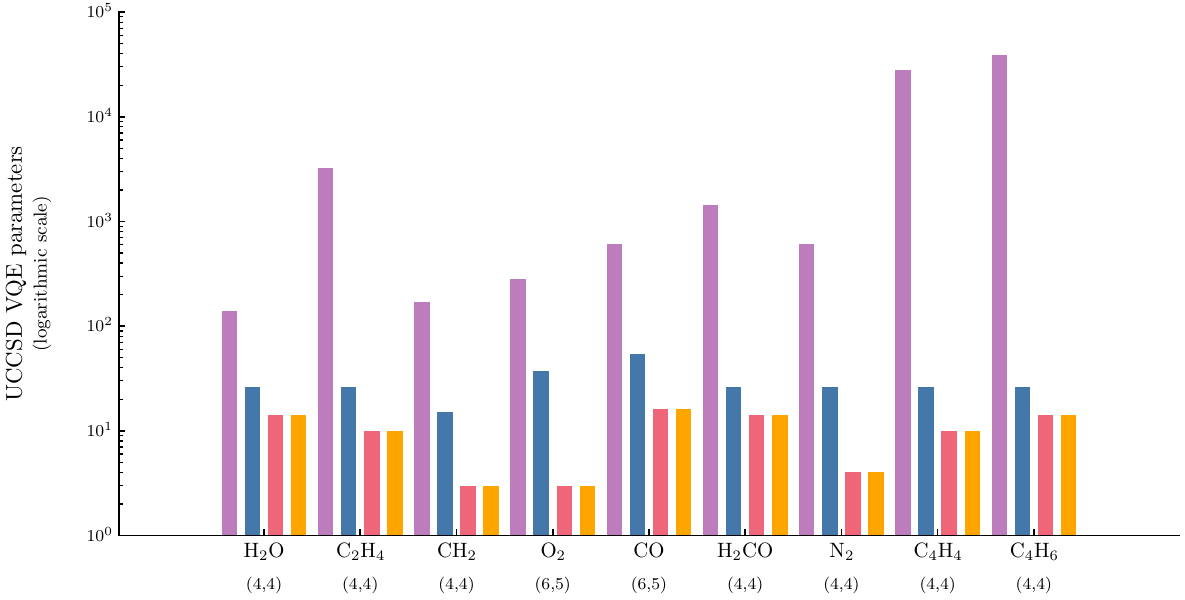}
  }\hfill
  \subfloat[HE-SCA VQE parameters\label{subfig:hea_params}]{
    \includegraphics[width=0.48\textwidth]{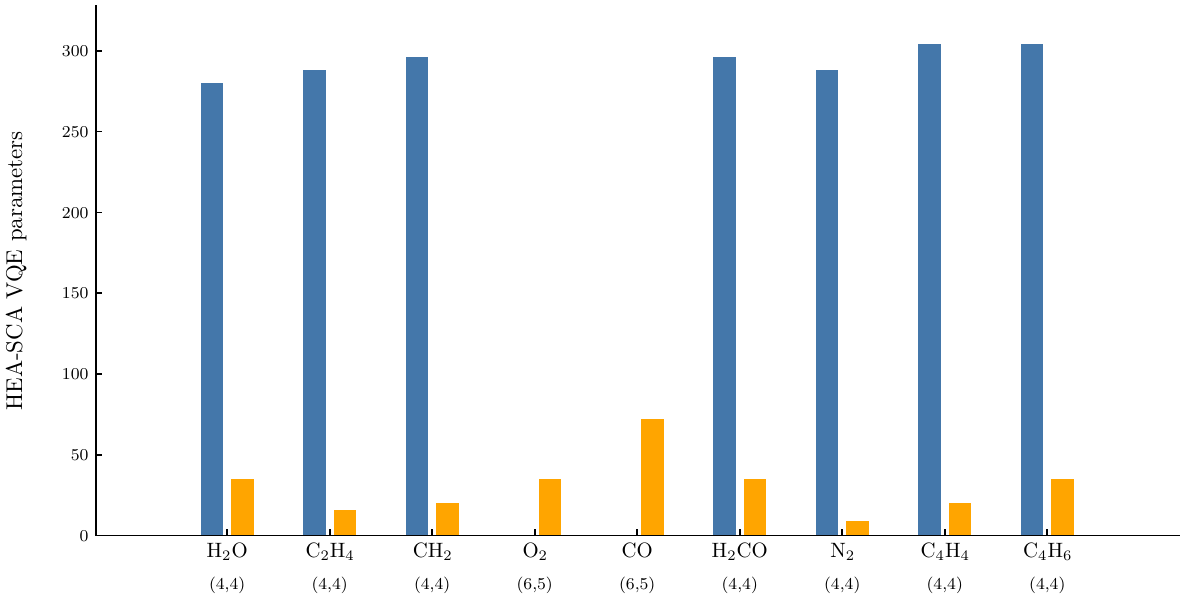}
  }\\[6pt]

  \hspace*{\fill}
  \subfloat[HE-SCA circuit layers\label{subfig:hea_layers}]{
    \includegraphics[width=0.48\textwidth]{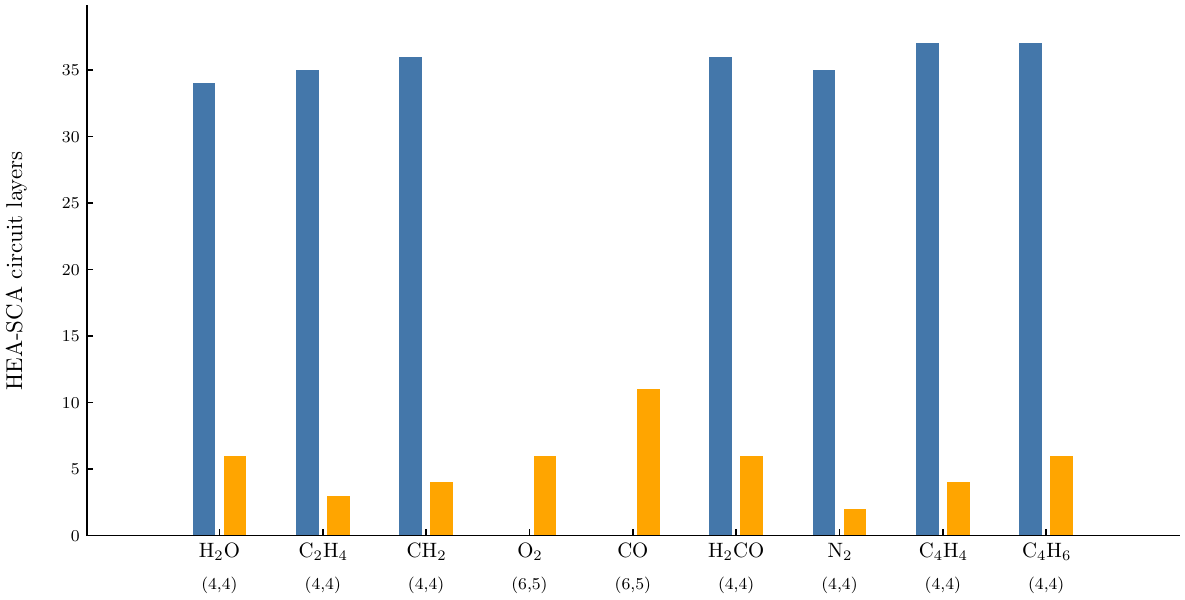}
  }
  \hspace*{\fill}

  \caption{%
    Circuit complexity metrics corresponding to Tables~\ref{tab:resources-part1}--\ref{tab:resources-part2}: circuit depth (UCCSD/HE-SCA), CNOT counts (UCCSD/HE-SCA),
    VQE parameter counts (UCCSD/HE-SCA), and HE-SCA circuit layers. SAE-CAS (\legendsquare{saecas}) yields shallower circuits,
    fewer entangling gates, and fewer parameters than JW (\legendsquare{jw}), JW-CAS (\legendsquare{jwcas}) and JW-CAS (SF) (\legendsquare{jwcassf}).
  }
  \label{fig:resources_b}
\end{figure*}

\begin{figure*}
  \centering
  \subfloat[Number of qubits\label{subfig:qubits}]{
    \includegraphics[width=0.48\textwidth]{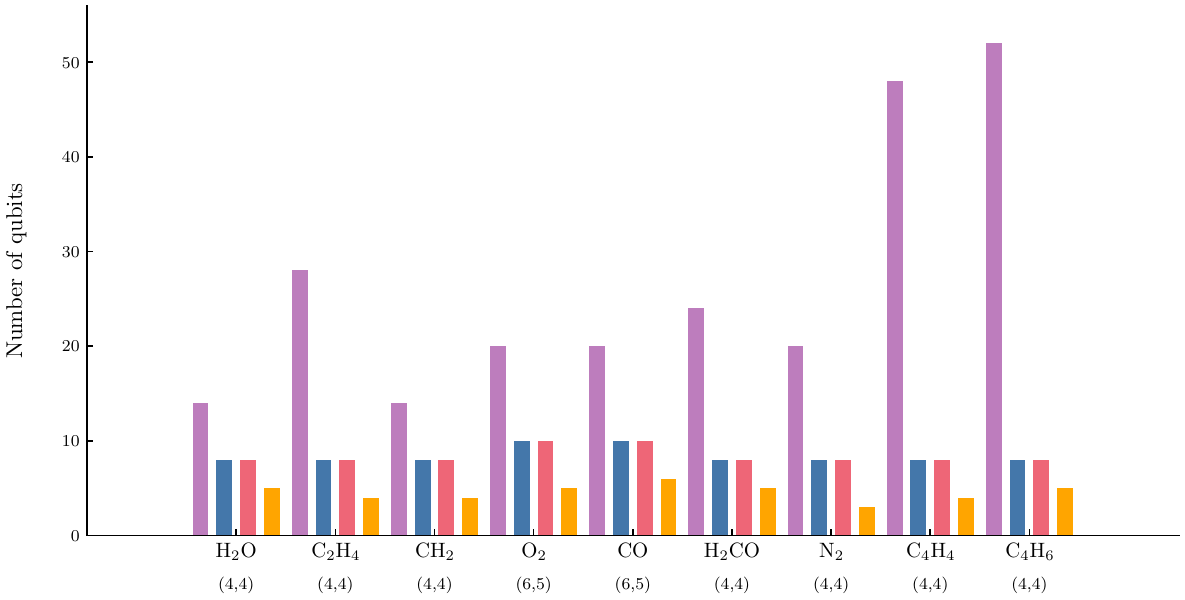}
  }\hfill
  \subfloat[Hamiltonian Pauli count\label{subfig:pauli_weight}]{
    \includegraphics[width=0.48\textwidth]{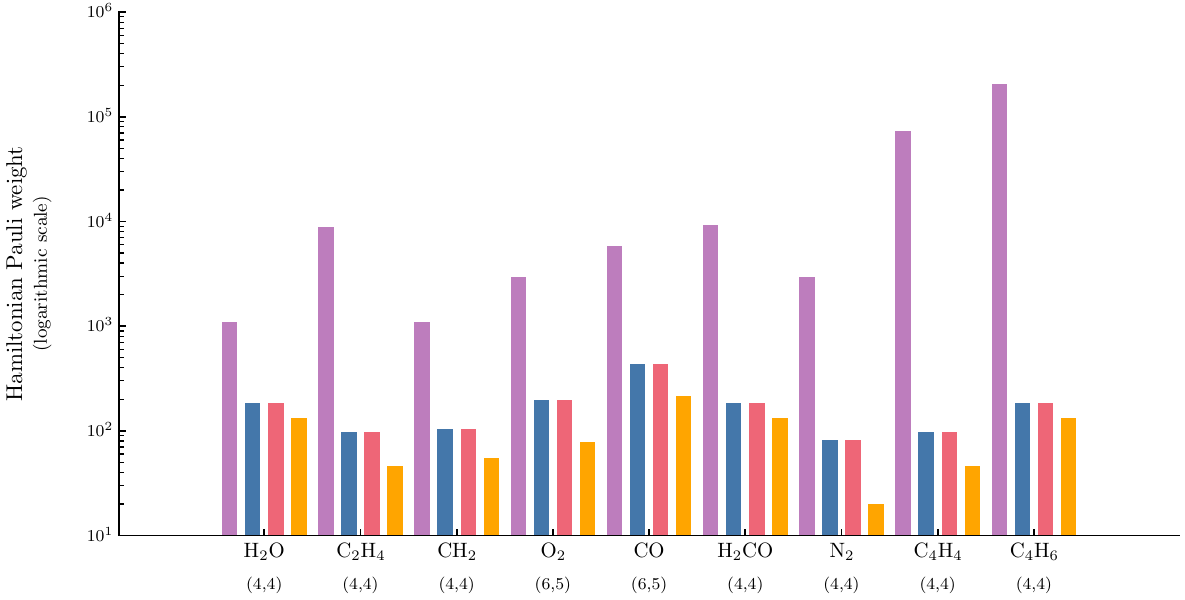}
  }
  \caption{%
    Qubit resource metrics corresponding to Tables~\ref{tab:resources-part1}--\ref{tab:resources-part2}: number of qubits and Hamiltonian Pauli count across all molecules in
    Tables~\ref{tab:resources-part1}--\ref{tab:resources-part2}. SAE-CAS (\legendsquare{saecas}) reduces both relative to JW (\legendsquare{jw}), JW-CAS (\legendsquare{jwcas}) and JW-CAS (SF) (\legendsquare{jwcassf}).
  }
  \label{fig:resources_a}
\end{figure*}

\begin{figure*}
  \centering

  \subfloat[UCCSD VQE iterations\label{subfig:uccsd_iters}]{
    \includegraphics[width=0.48\textwidth]{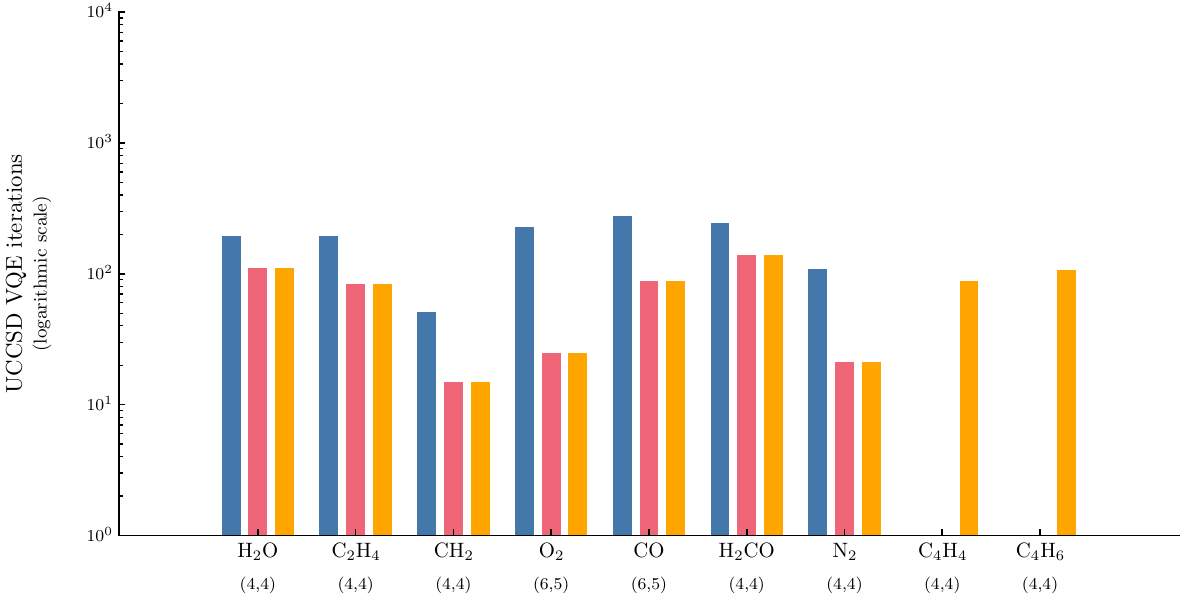}
  }\hfill
  \subfloat[HE-SCA VQE iterations\label{subfig:hea_iters}]{
    \includegraphics[width=0.48\textwidth]{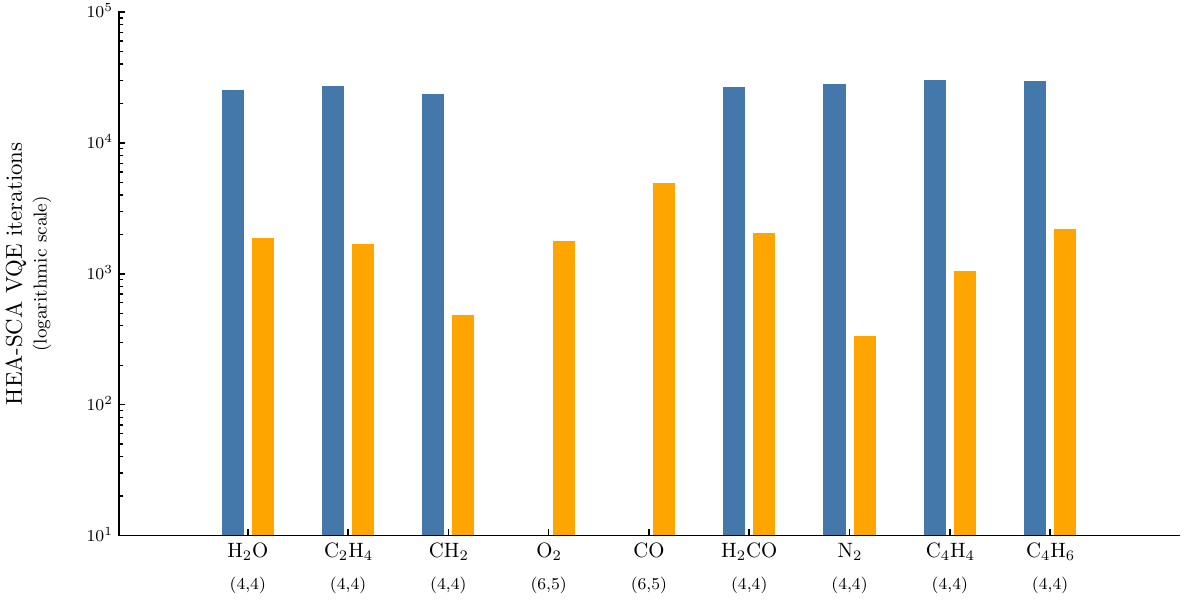}
  }\\[6pt]

  \subfloat[UCCSD VQE energy error ($E_h$)\label{subfig:uccsd_err}]{
    \includegraphics[width=0.48\textwidth]{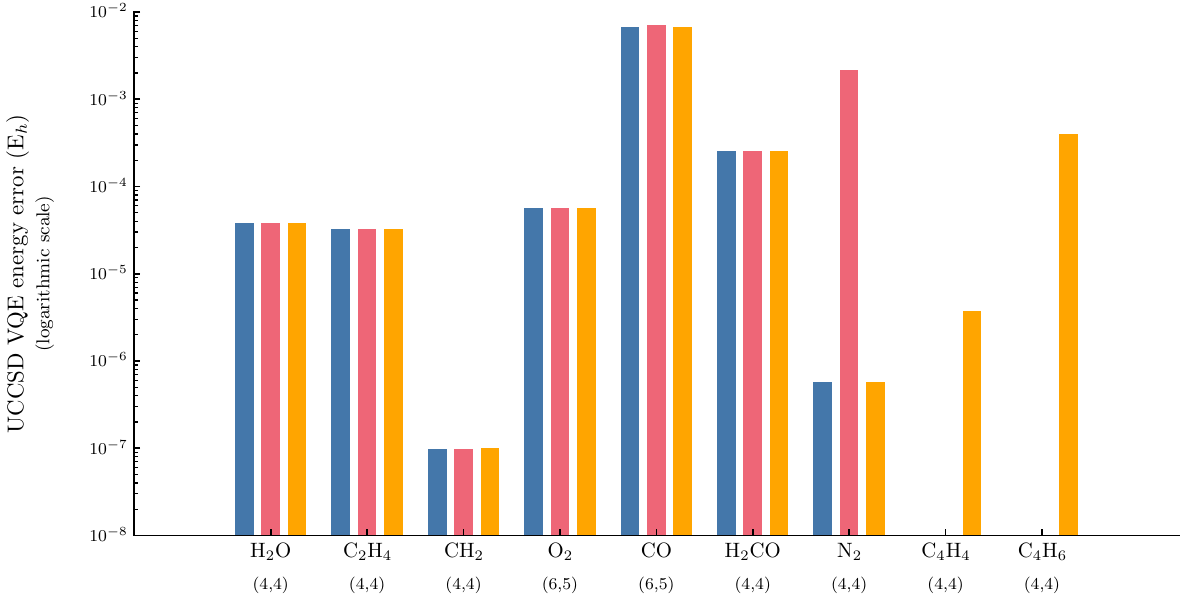}
  }\hfill
  \subfloat[HE-SCA VQE energy error ($E_h$)\label{subfig:hea_err}]{
    \includegraphics[width=0.48\textwidth]{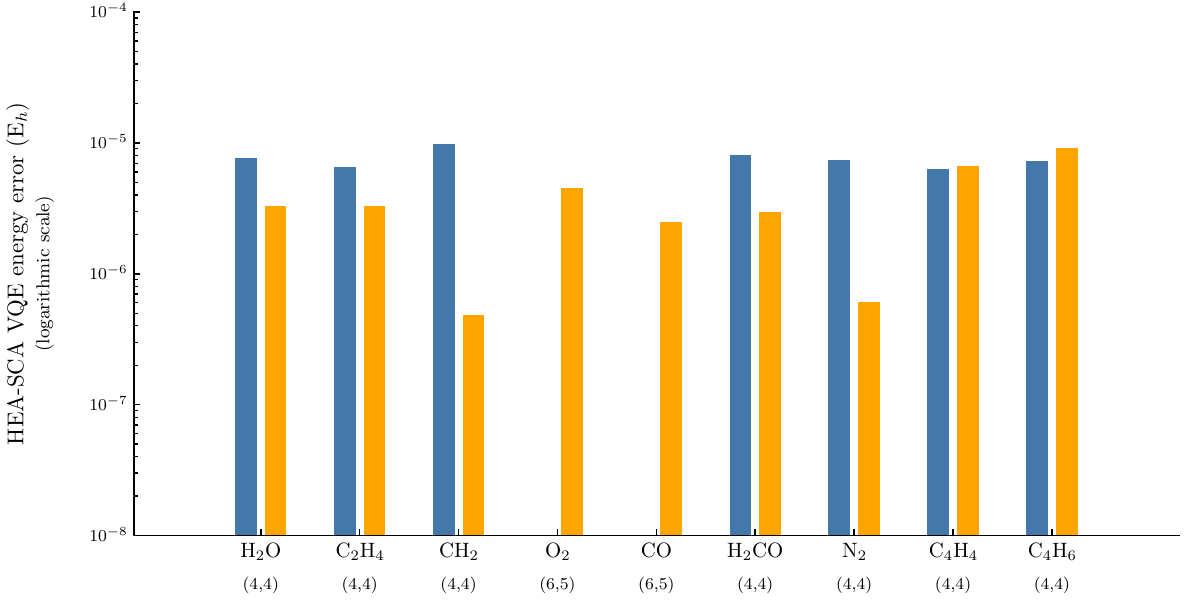}
  }

  \caption{%
    Optimisation and accuracy metrics corresponding to Tables~\ref{tab:resources-part1}--\ref{tab:resources-part2}.
    Panels: (a,b) VQE iteration counts for UCCSD and HE-SCA; (c,d) absolute VQE energy errors relative to the CAS energy calculated in \textsf{PySCF} in units of Hartree.
    Encoding symmetry at the mapping level with SAE-CAS (\legendsquare{saecas}) reduces iteration counts with respect to JW-CAS (\legendsquare{jwcas}) and JW-CAS (SF) (\legendsquare{jwcassf}) and allows us to attain CAS energies with HE-SCA in cases where the unreduced JW-CAS ansatz did not converge within the same budgets.
  }
  \label{fig:performance}
\end{figure*}

\pagebreak

\appendix

\section{Water molecule example}
\label{app:water_example}

The water molecule (\ce{H2O}) has point group symmetry $C_{2v}$, with symmetries $E$ (the identity), $C_2$ (the rotation by $\pi$ around the $z$ axis), $\sigma_v$ (the reflection across the $xz$-plane) and $\sigma_v'$ (the reflection across the $yz$-plane):

\begin{center}
  \includegraphics[width=0.15\linewidth]{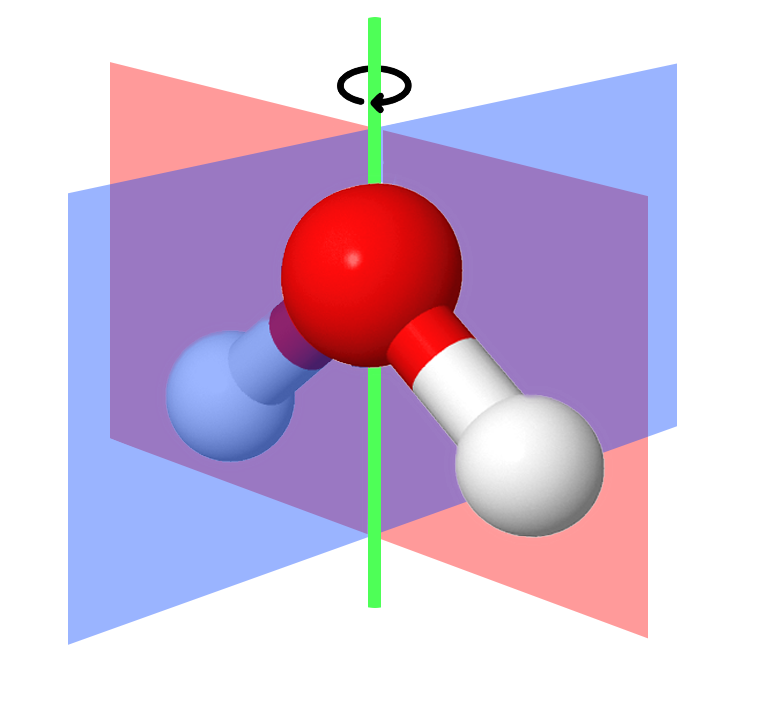}
\end{center}

The corresponding four one-dimensional irreps are $A_1$, $A_2$, $B_1$ and $B_2$, defined by the character table of $C_{2v}$:

\begin{center}
  \begin{tabular}{ c|c c c c}
    & \cellcolor{brown!10} $E$ & \cellcolor{green!10} $C_2$ & \cellcolor{blue!10} $\sigma_v$ & \cellcolor{red!10} $\sigma_v'$ \\ 
    \hline
      \cellcolor{lime!20}$A_1$ & $1$ & $1$ & $1$ & $1$ \\
      \cellcolor{cyan!20}$A_2$ & $1$ & $1$ & $-1$ & $-1$ \\
      \cellcolor{purple!20}$B_1$ & $1$ & $-1$ & $1$ & $-1$ \\
      \cellcolor{orange!20}$B_2$ & $1$ & $-1$ & $-1$ & $1$ \\
  \end{tabular}
\end{center}

Using symmetry-adapted molecular orbitals each orbital lies in a $C_{2v}$ irreducible representation. Then the qubit form of point-group symmetries in the Jordan--Wigner basis is particularly simple: each generator maps to a tensor product of $Z$’s acting exactly on the spin--orbitals that are antisymmetric under it. For water in a minimal basis, the 7 symmetry-adapted spatial orbitals correspond to 14 spin-orbitals, labelled by their irreducible representation. Each point-group element maps to products of $Z$’s on the Jordan--Wigner qubits in a way determined by the point-group character table and by the corresponding orbital's irreducible representation. Similarly, $P_{\alpha}$ and $P_{\beta}$ act as $Z$ operators on the qubits corresponding to spin-orbitals with spin $\alpha$ and $\beta$ respectively. This is shown for the example in the table below:

\begin{center}
\begin{tabular}{cccccccccccccccc}
     & & \multicolumn{2}{c}{\includegraphics[width=0.08\linewidth]{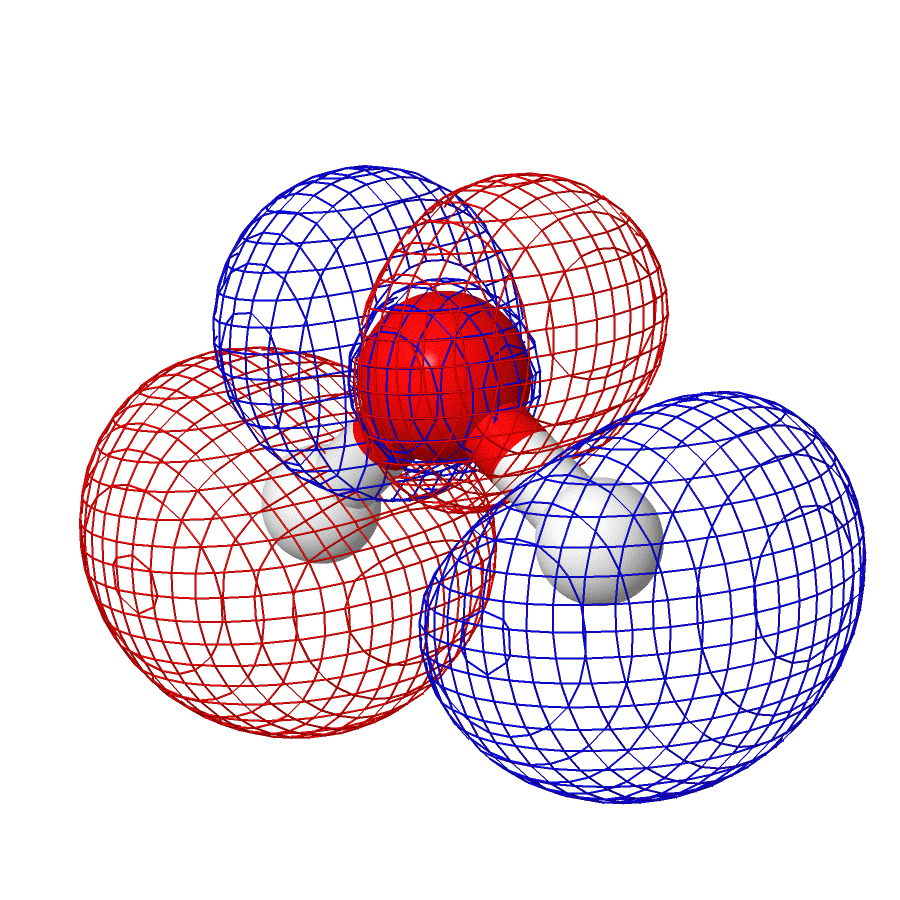}} & \multicolumn{2}{c}{\includegraphics[width=0.09\linewidth]{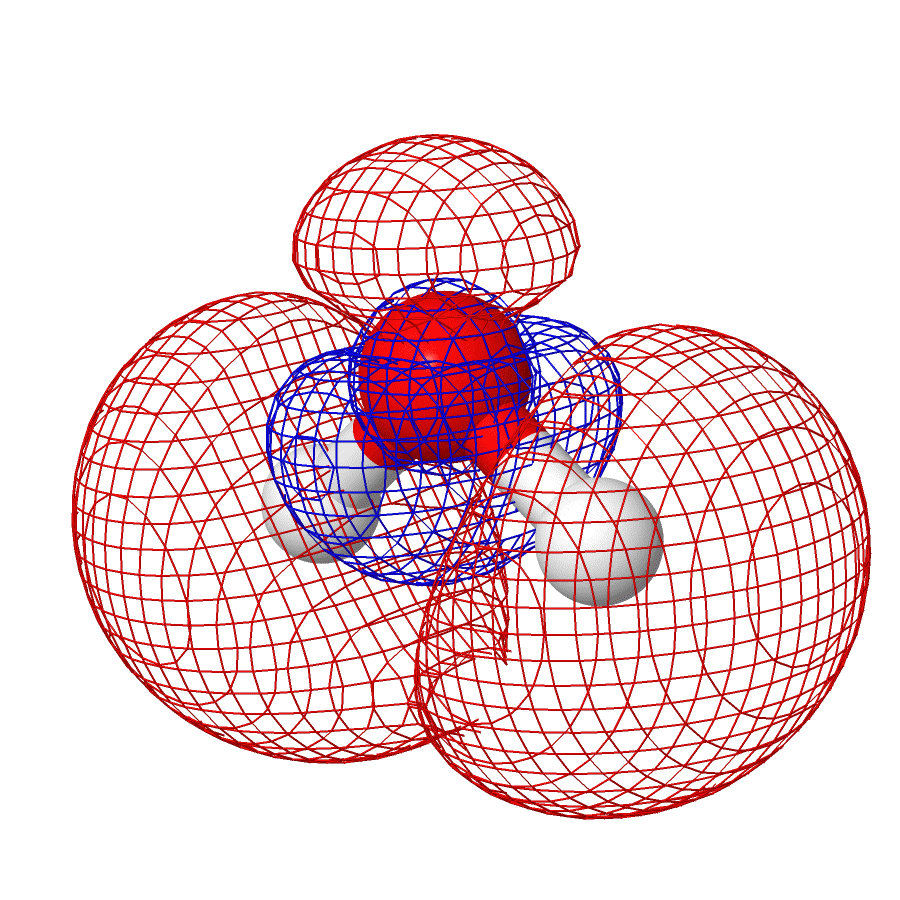}} & \multicolumn{2}{c}{\includegraphics[width=0.09\linewidth]{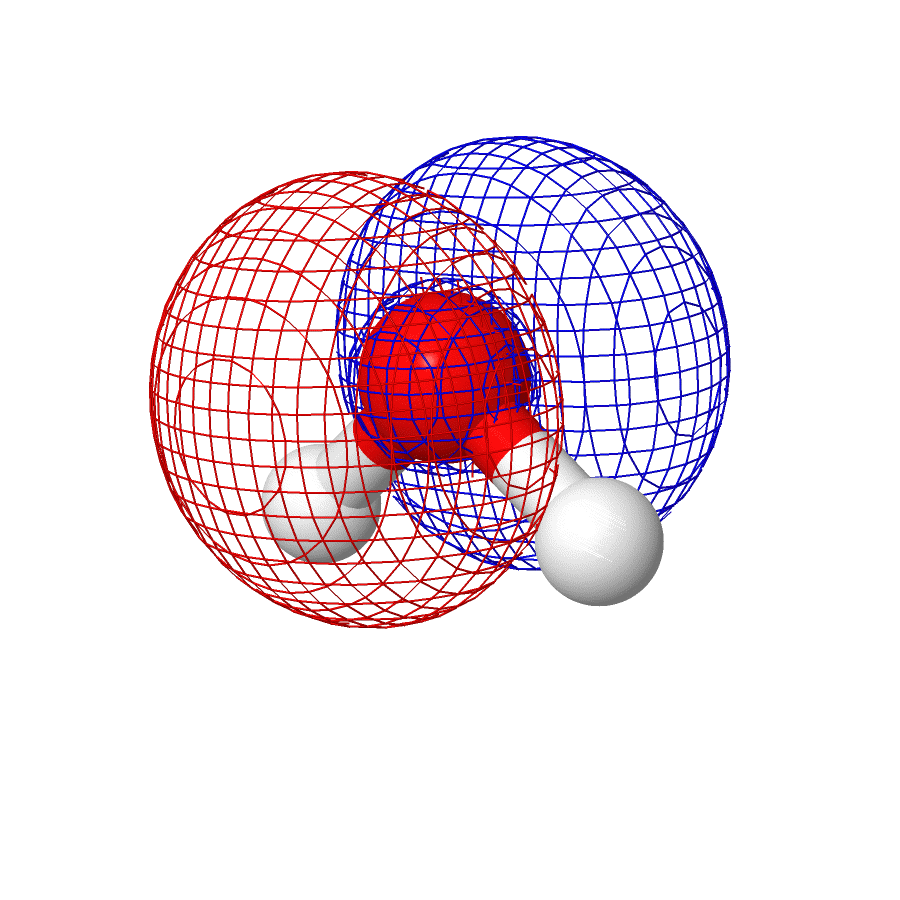}} & \multicolumn{2}{c}{\includegraphics[width=0.09\linewidth]{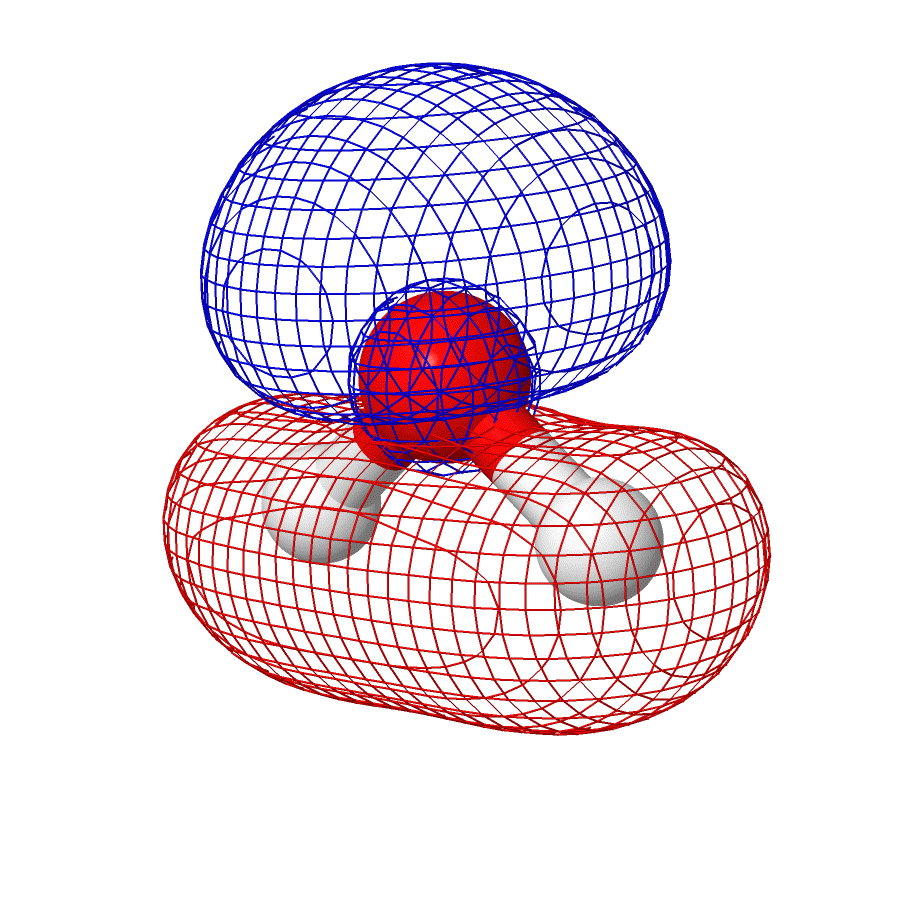}} & \multicolumn{2}{c}{\includegraphics[width=0.09\linewidth]{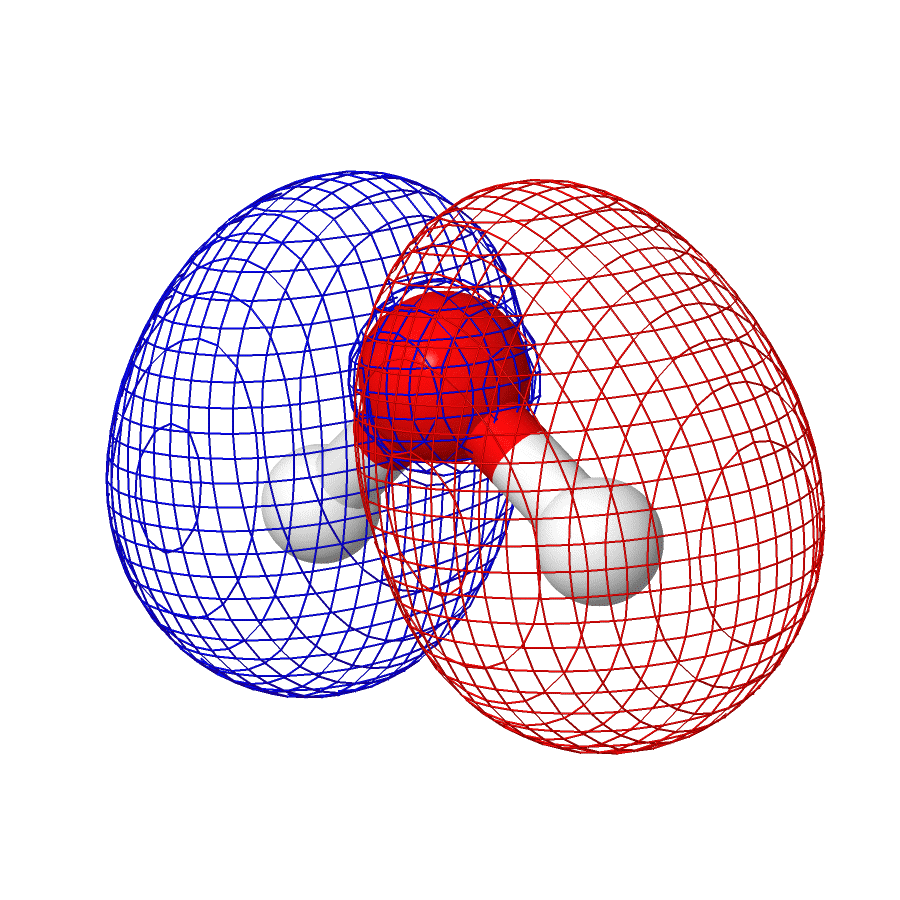}} & \multicolumn{2}{c}{\includegraphics[width=0.09\linewidth]{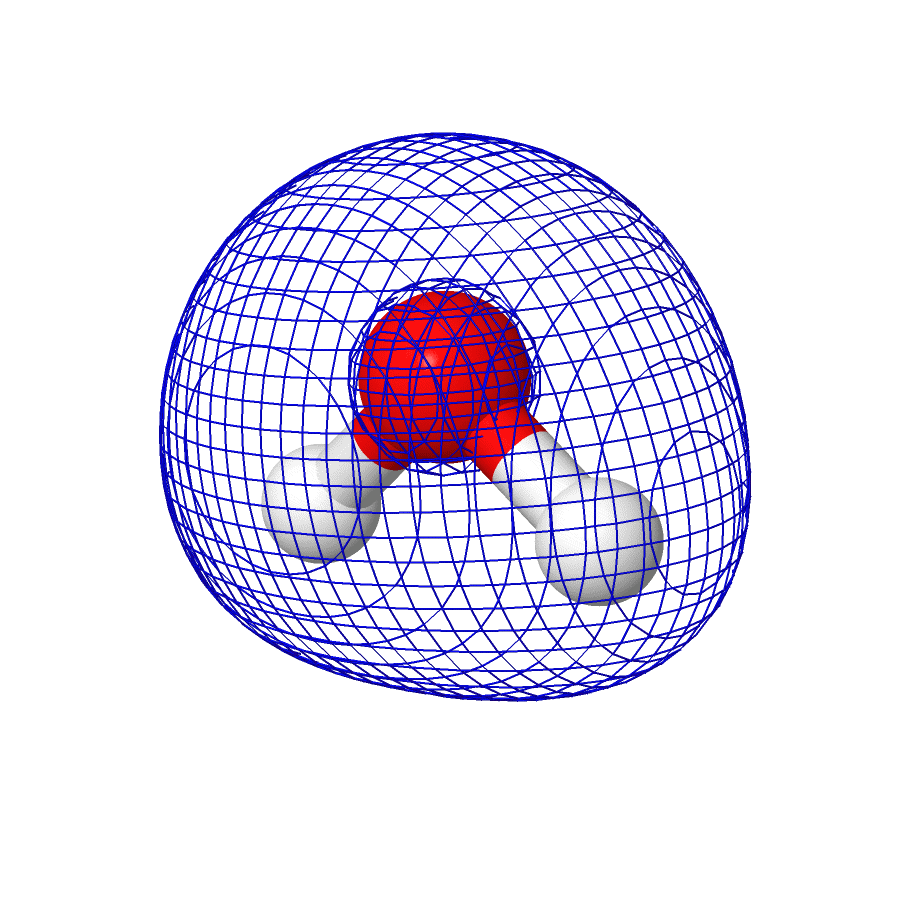}} & \multicolumn{2}{c}{\includegraphics[width=0.09\linewidth]{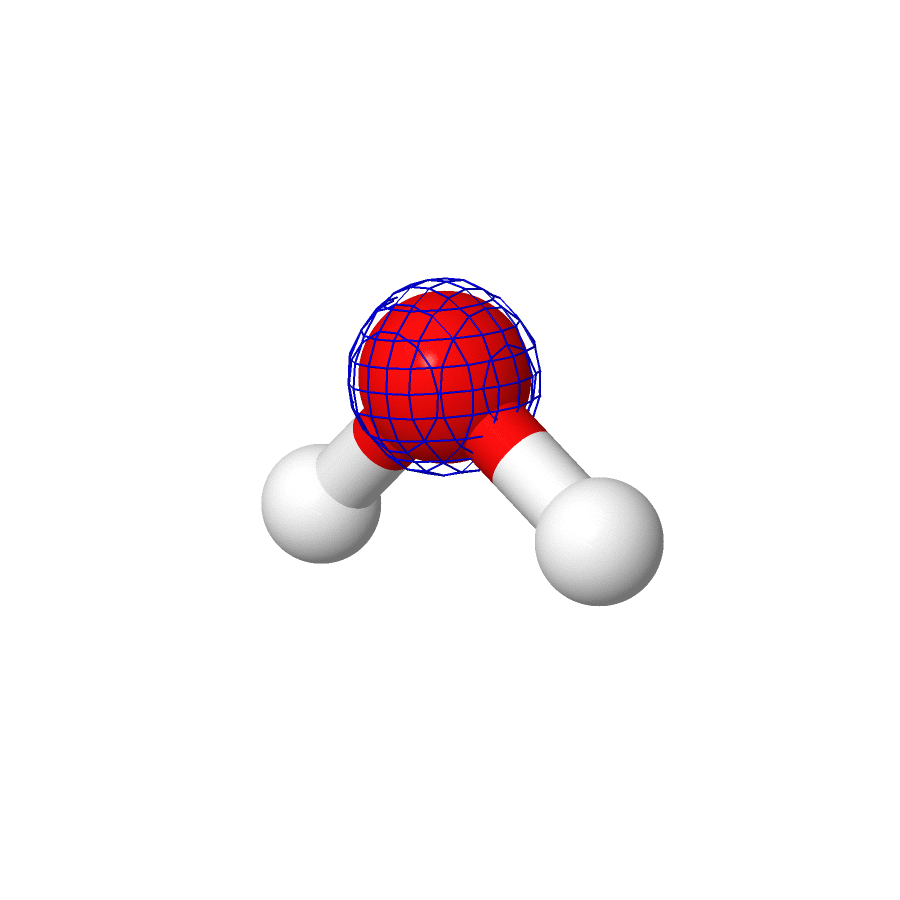}}
     \\
          & & \multicolumn{2}{c}{\cellcolor{orange!20} $2b_2$} 
          & \multicolumn{2}{c}{\cellcolor{lime!20} $4a_1$} 
          & \multicolumn{2}{c}{\cellcolor{purple!20} $1b_1$} 
          & \multicolumn{2}{c}{\cellcolor{lime!20} $3a_1$} 
          & \multicolumn{2}{c}{\cellcolor{orange!20} $1b_2$} 
          & \multicolumn{2}{c}{\cellcolor{lime!20} $2a_1$} 
          & \multicolumn{2}{c}{\cellcolor{lime!20} $1a_1$}
     \\
     & & \footnotesize $\beta$ 
     & \footnotesize $\alpha$
     & \footnotesize $\beta$ 
     & \footnotesize $\alpha$
     & \footnotesize $\beta$ 
     & \footnotesize $\alpha$
     & \footnotesize $\beta$ 
     & \footnotesize $\alpha$
     & \footnotesize $\beta$ 
     & \footnotesize $\alpha$
     & \footnotesize $\beta$ 
     & \footnotesize $\alpha$
     & \footnotesize $\beta$ 
     & \footnotesize $\alpha$
     \\
     \cellcolor{green!10} $C_2$ & $\rightarrow$ &  $Z_{13}$ &  $Z_{12}$ & & &  $Z_{9}$ &  $Z_{8}$ & & &  $Z_{5}$ &  $Z_{4}$ & & & & \\
     \cellcolor{blue!10} $\sigma_v$ & $\rightarrow$ &  $Z_{13}$ &  $Z_{12}$ & & & & & & &  $Z_{5}$ &  $Z_{4}$ & & & & \\
     \cellcolor{red!10} $\sigma_v'$ & $\rightarrow$  & & & & &  $Z_{9}$ &  $Z_{8}$ & & & & & & & & \\ \\
    $P_{\alpha}$ & $\rightarrow$  & &  $Z_{12}$ & &   $Z_{10}$ & &  $Z_{8}$ & &  $Z_{6}$ & &  $Z_{4}$ & &  $Z_{2}$ & &  $Z_{0}$ \\ 
    $P_{\beta}$ & $\rightarrow$  &  $Z_{13}$ & &   $Z_{11}$ & &  $Z_{9}$ & &  $Z_{7}$ & &  $Z_{5}$ & &  $Z_{3}$ & &  $Z_{1}$ &
\end{tabular}
\end{center}

The CAS(4, 4) active space selection corresponds to freezing three of those seven molecular orbitals, and can be understood as introducing additional approximate symmetries each acting as a one-qubit $Z$.

\section{Properties of affine Clifford maps}
\label{app:clifford-maps}

We work over the field $\mathbb F_2$. Bitstrings are column vectors $a \in\mathbb F_2^n$. The standard basis vectors are $e_j$ for $j=1,\dots,n$. For bitstrings $x, z \in\mathbb F_2^n$ we define the multi-qubit Paulis:
\begin{equation}\label{eq:def-Paulis}
X(x):=\bigotimes_{k=1}^n X_k^{x_k},
\qquad
Z(z):=\bigotimes_{k=1}^n Z_k^{z_k}.
\end{equation}
They act on computational basis states $\ket a$ by
\begin{align}
\label{eq:x-action}
X(x)\ket a &= \ket{a\oplus x},
\\
\label{eq:z-action}
Z(z)\ket a &= (-1)^{z\cdot a}\ket a,
\end{align}
where $v\cdot a$ is the standard dot product over $\mathbb F_2$.
In particular, $X_j:=X(e_j)$ and $Z_j:=Z(e_j)$.

\begin{proposition}[Pauli permutations are Pauli $X$]
\label{prop:pauli-permutation-test}
Any $n$-qubit Pauli can be written as $\omega\,X(x)Z(z)$ with $\omega\in\{\pm1,\pm i\}$ and $x,z\in\mathbb{F}_2^n$.
Its nonzero computational-basis matrix elements satisfy
\begin{equation}
\label{eq:pauli-matrix-elements}
\big\langle a\oplus x \,\big|\, \omega X(x)Z(z) \,\big|\, a \big\rangle
= \omega\,(-1)^{z\cdot a},
\end{equation}
so it is a signed permutation matrix up to a global phase. Moreover, it is a permutation (all nonzero matrix elements are $+1$) if and only if $z=0$ and $\omega=1$, in which case the operator is $X(x)$.
\end{proposition}

\begin{proof}
Because operators acting on different qubits commute and on each qubit we have $Y = iXZ$, any Pauli can be written as $\omega\,X(x)Z(z)$ with $\omega \in \{\pm1,\pm i\}$.
Then
\begin{align}
\omega X(x)Z(z)\ket{a}
&= \omega\,(-1)^{z\cdot a}\ket{a} && \text{by \eqref{eq:z-action}}\\
&= \omega\,(-1)^{z\cdot a}\ket{a\oplus x} && \text{by \eqref{eq:x-action}}\\
\end{align}
which yields \eqref{eq:pauli-matrix-elements} and shows there is exactly one nonzero matrix element per row and per column (a permutation up to a phase).
If $z\neq 0$, choose $a$ with $z\cdot a=1$ to obtain a $-\omega$ entry; if $\omega\in\{-1,\pm i\}$, some nonzero matrix element do not equal $1$.
Thus all nonzero entries are $+1$ exactly when $z=0$ and $\omega=1$, in which case the operator is $X(x)$.
\end{proof}

\begin{proposition}[CNOTs generate linear maps on bitstrings]\label{lem:CNOT-generate-GL}
For $i\neq j$, let $T_{j,i}:=I+E_{j,i}$, the transvection that sends $x_j\mapsto x_j\oplus x_i$ (a CNOT with control $i$ and target $j$). Then $\langle T_{j,i}:i\neq j\rangle = GL(n,2)$. Equivalently, every $A\in GL(n,2)$ can be written as a product of $T_{j,i}$’s.
\end{proposition}

\begin{proof}
Over $\mathbb F_2$ there is no nontrivial scaling, so Gauss--Jordan elimination uses only row additions and row swaps. Left-multiplication by $T_{j,i}$ adds row $i$ to row $j$, and a swap is a product of three transvections by the identity
\begin{equation}\label{eq:xor-swap}
\mathrm{swap}_{i,j}=T_{i,j}\,T_{j,i}\,T_{i,j}.
\end{equation}
Corresponding to the well-known quantum circuit identity
\begin{equation}
\begin{quantikz} 
& \gate[wires=2]{\text{SWAP}} & \gate[wires=2,style={draw=none,fill=none}]{=} & \ctrl{1} & \targ{} & \ctrl{1} & \qw \\ 
& & & \targ{} & \ctrl{-1} & \targ{} & \qw 
\end{quantikz}
\end{equation}
Thus there exist transvections $R_1,\ldots,R_m$ with $R_m\cdots R_1A=I$. Taking inverses and using $T_{j,i}^{-1}=T_{j,i}$ gives
$A=R_1\cdots R_m$, a product of transvections, as claimed.
\end{proof}

\begin{proposition}[X and CNOTs generate affine maps on bitstrings]\label{thm:affine-equals-XCNOT}
A unitary $U$ acts as $\ket x\mapsto \ket{A x\oplus b}$ with $A\in GL(n,2)$ and $b\in\mathbb F_2^n$
if and only if $U$ is a circuit over $X$ and $\mathrm{CNOT}$.
\end{proposition}

\begin{proof}
\emph{Only if.} Each generator is affine: $X_k$ adds $e_k$, and $\mathrm{CNOT}_{i\to j}$ applies $T_{j,i}$.
Composing these gives $\ket{A x\oplus b}$ with $A\in GL(n,2)$.

\emph{If.} By Lemma~\ref{lem:CNOT-generate-GL}, write $A$ as a product of additions and realize each by a $\mathrm{CNOT}$.
Append $X_k$ for each $b_k=1$ to implement $x\mapsto x\oplus b$.
\end{proof}

\begin{proposition}[Permutation Cliffords are affine maps on bitstrings]\label{thm:perm-Clifford-affine}
An $n$-qubit unitary $U$ is a Clifford that permutes the computational basis
if and only if $U$ is affine (hence $U\in\langle X,\mathrm{CNOT}\rangle$).
\end{proposition}

\begin{proof}
\emph{Forward direction.}
Suppose $U$ is a permutation matrix: $U\ket x=\ket{\pi(x)}$.
Then
\begin{equation}\label{eq:conj-perm-X}
U X(e_j) U^\dagger=\sum_x \ket{\pi(x\oplus e_j)}\!\bra{\pi(x)} .
\end{equation}
This is a permutation matrix with all nonzero entries $+1$. Since $U$ is Clifford, the left-hand side is a Pauli,
so by Proposition~\ref{prop:pauli-permutation-test} it must equal $X(u_j)$ for some $u_j$.
Reading off matrix entries gives $\pi(x\oplus e_j)=\pi(x)\oplus u_j$ for all $x,j$.
Let $b=\pi(0)$ and define $f(x)=\pi(x)\oplus b$. Then $f(x\oplus e_j)=f(x)\oplus f(e_j)$, hence by additivity on generators
\begin{equation}\label{eq:perm-linear-part}
f(x)=\sum_{j=1}^n x_j f(e_j)=A x .
\end{equation}
Thus $\pi(x)=A x\oplus b$; since $\pi$ is bijective, $A\in GL(n,2)$.

\emph{Reverse direction.}
If $U\ket x=\ket{A x\oplus b}$ with $A\in GL(n,2)$, then $x\mapsto A x\oplus b$ is a bijection,
so $U$ permutes the basis. By Theorem~\ref{thm:affine-equals-XCNOT}, $U\in\langle X,\mathrm{CNOT}\rangle$.
\end{proof}

\section{Tableau of an affine Clifford map}
\label{app:clifford-tableau}

\begin{proposition}[Tableau of an affine Clifford map]\label{prop:block-tableau}
Let $T\in GL(n,2)$ and $b\in\mathbb F_2^n$, and define the unitary
\begin{equation}\label{eq:affine-U}
C\ket a=\ket{T a\oplus b},\qquad a\in\mathbb F_2^n .
\end{equation}
Write the Clifford tableau in block form
\begin{equation}\label{eq:tableau-blocks}
M=\begin{pmatrix} M_{ZZ} & M_{ZX} \\ M_{XZ} & M_{XX} \end{pmatrix},\qquad
s=\binom{s_Z}{s_X}\in\mathbb F_2^{2n},
\end{equation}
meaning that for each $j=1,\dots,n$ (with $e_j$ the $j$th standard basis vector),
\begin{align}
\label{eq:tableau-definition-z}
C\,Z_j\,C^\dagger&=(-1)^{(s_Z)_j}\,Z\!\big(M_{ZZ}e_j\big)\,X\!\big(M_{XZ}e_j\big)\\
\label{eq:tableau-definition-x}
C\,X_j\,C^\dagger&=(-1)^{(s_X)_j}\,Z\!\big(M_{ZX}e_j\big)\,X\!\big(M_{XX}e_j\big).
\end{align}
Then 
\begin{equation}\label{eq:tableau-result}
M_{ZZ}=(T^{-1})^{\!T}, \qquad M_{XX}=T, \qquad M_{ZX} = M_{XZ}=0
\qquad s_Z=T^{-1}b,\quad s_X=0 .
\end{equation}
\end{proposition}

\begin{proof}From \eqref{eq:affine-U} we immediately have
\begin{equation}\label{eq:clifford-proof1}
C\ket{a}=\ket{T a\oplus b}.
\end{equation}
Inverting the bijection $q=T a\oplus b$ gives
\begin{equation}\label{eq:clifford-proof4}
C^\dagger\ket{q}=\ket{T^{-1}(q\oplus b)} .
\end{equation}

\emph{Pauli X.} For any $x\in\mathbb F_2^n$,
\begin{equation}\label{eq:clifford-proof6}
X(x)\ket{q}=\ket{q\oplus x}.
\end{equation}
Thus, for any $a$,
\begin{equation}
\begin{aligned}
C\,X(x)\,C^\dagger\ket a
&= C\!\left(X(x)\ket{T^{-1}(a\oplus b)}\right) && \text{by \eqref{eq:clifford-proof4}}\\
&= C\ket{T^{-1}(a\oplus b)\oplus x} && \text{by \eqref{eq:clifford-proof6}}\\
&= \ket{T\!\big(T^{-1}(a\oplus b)\oplus x\big)\oplus b} && \text{by \eqref{eq:clifford-proof1}}\\
&= \ket{(a\oplus b)\oplus T x\oplus b}\\
&= \ket{a\oplus T x}\\
&= X(Tx)\ket a,
\end{aligned}
\end{equation}
so
\begin{equation}\label{eq:X-conjugates}
C\,X(x)\,C^\dagger=X(Tx).
\end{equation}
Taking $x=e_j$ and comparing with \eqref{eq:tableau-definition-x} yields $M_{XX}e_j=Te_j$, $M_{ZX}e_j=0$, and $(s_X)_j=0$ in \eqref{eq:tableau-result}.

\emph{Pauli Z.} For any $z\in\mathbb F_2^n$ and any $q$,
\begin{equation}\label{eq:clifford-proof2}
Z(z)\ket{q}=(-1)^{q\cdot z}\ket{q}.
\end{equation}
Hence, for any $a$,
\begin{equation}
\begin{aligned}
C\,Z(z)\,C^\dagger\ket a
&= C\!\left(Z(z)\ket{T^{-1}(a\oplus b)}\right) && \text{by \eqref{eq:clifford-proof4}}\\
&= (-1)^{\,z\cdot T^{-1}(a\oplus b)}\,C \ket{T^{-1}(a\oplus b)} && \text{by \eqref{eq:clifford-proof2}}\\
&= (-1)^{\,z\cdot T^{-1}(a\oplus b)}\,\ket a\\
&= (-1)^{\,z\cdot T^{-1}b}\,(-1)^{\,z\cdot T^{-1}a}\ket a\\
&= (-1)^{\,z\cdot T^{-1}b}\,(-1)^{\,((T^{-1})^{\!T}z)\cdot a}\ket a && \text{by \eqref{eq:clifford-proof7}}\\
&= (-1)^{\,z\cdot T^{-1}b}\,Z\!\big((T^{-1})^{\!T}z\big)\ket a,
\end{aligned}
\end{equation}
where we have used the transpose property of the (binary) inner product:
\begin{equation}
\label{eq:clifford-proof7}
A^{-1} x  \cdot y = x \cdot (A^{-1})^{T} y,
\end{equation}
so
\begin{equation}\label{eq:Z-conjugates}
C\,Z(z)\,C^\dagger=(-1)^{z\cdot T^{-1}b}\,Z\!\big((T^{-1})^{\!T}z\big).
\end{equation}
Taking $z=e_j$ and comparing with \eqref{eq:tableau-definition-z} yields $M_{ZZ}e_j=(T^{-1})^{\!T}e_j$, $M_{XZ}e_j=0$, and $(s_Z)_j=(T^{-1}b)_j$.

Collecting the columns over $j=1,\dots,n$ we recover \eqref{eq:tableau-result}.
\end{proof}

\section{Affine Clifford for Boolean symmetries}
\label{app:sae-clifford}

\begin{proposition}[Affine Clifford for Boolean symmetries]
Let $g_1,\dots,g_k$ be $k$ independent Boolean symmetry generators acting diagonally on Jordan--Wigner computational-basis states of an $n$-qubit register. Choose them so that, for each $i\in[k]$, $g_i$ has eigenvalue $-1$ on qubit $i$, while $g_j$ has eigenvalue $+1$ there for all $j\neq i$. Let $S\in\{0,1\}^{k\times n}$ be the orbital-character matrix of the generators $g_i$.

Then any Clifford $C$ that
\begin{enumerate}[label=(\roman*)]
\item maps computational-basis states to computational-basis states without phases;
\item acts as the identity on the last $n-k$ qubits; and
\item maps, by conjugation, each $g_i$ to a single-qubit $Z$ operator on qubit $i$, up to a phase;
\item maps each of the first $k$ bits to $0$ on every joint eigenstate of all $g_i$ with eigenvalues $(-1)^{c_i}$ for some $c\in\{0,1\}^k$;
\end{enumerate}
is unique up to the choice and ordering of generators, and acts on bitstrings as
\begin{equation}
\label{eq:clifford-affine-action}
C\ket{a}=\ket{T a \oplus b},
\end{equation}
where $a\in\{0,1\}^n$, all arithmetic is over $\mathbb{F}_2$, and
\begin{enumerate}
\item the linear map $T\in\{0,1\}^{n\times n}$ is obtained from the identity $I_n$ by replacing its first $k$ rows with $S$;
\item the shift $b\in\{0,1\}^n$ is $b=\begin{pmatrix} c\\0\end{pmatrix}$.
\end{enumerate}
\end{proposition}

\begin{proof}
Since each $g_i$ acts diagonally in the JW computational basis and its action commutes with a reordering of the spin-orbitals, its JW image is a $Z$--string:
\begin{equation}
\Phi_{\mathrm{JW}}(g_i)=Z(z_i),\qquad z_i \in\mathbb F_2^n.
\end{equation}

\emph{By (i).} By Proposition~\ref{thm:perm-Clifford-affine}, if $C$ permutes computational-basis states without phases it is an affine Clifford:
\begin{equation}
\label{eq:affine-form-in-proof}
C\ket{a}=\ket{T a\oplus b},\qquad T\in GL(n,2),\ b\in\mathbb F_2^n .
\end{equation}

\emph{By (ii).} Acting as the identity on the last $n-k$ qubits means 
\begin{equation}
T_{ij} = \delta_{ij}, \qquad b_i = 0 \qquad (k < i \le n)
\label{eq:boolean-cliff-result1}
\end{equation}

\emph{By (iii) and (iv).} By Proposition~\ref{prop:block-tableau}, the action of $C$ on $Z$--strings is
\begin{equation}
\label{eq:Z-conj-again}
C Z(z_i) C^\dagger  =  (-1)^{z_i \cdot T^{-1} b} Z \big((T^{-1})^{T} z_i\big).
\end{equation}
Condition \emph{(iii)} and \emph{(iv)} together then require for $1 \le i \le k$:
\begin{align}
(-1)^{c_i} Z_{i} &= (-1)^{z_i \cdot T^{-1} b} Z \big((T^{-1})^{T} z_{i}\big)\\
z_i &= S^{T}e_i,
\end{align}
This gives
\begin{equation}
(-1)^{c_i} Z_{i} = (-1)^{(S^{T}e_i) \cdot (T^{-1} b)} Z \big((T^{-1})^{T} S^{T}e_i\big),
\end{equation}
and hence the two conditions on $T$ and $b$:
\begin{align}
\label{eq:boolean-cliff-eq1}
({T^{-1}})^T S^T e_i &= e_i\\
\label{eq:boolean-cliff-eq2}
(S^{T}e_i) \cdot (T^{-1} b) &= c_i
\end{align}
From \eqref{eq:boolean-cliff-eq1} we get 
\begin{equation}
S^T e_i = T^T e_i \qquad (1 \le i \le k)
\label{eq:boolean-cliff-result2}
\end{equation}
and hence the first $k$ rows of $T$ and $S$ are equal.

\eqref{eq:boolean-cliff-eq2} is equivalent to
\begin{equation}
((T^{-1})^TS^{T}e_i) \cdot b = c_i
\label{eq:boolean-cliff-eq3}
\end{equation}
\eqref{eq:boolean-cliff-eq1} and \eqref{eq:boolean-cliff-eq3} then give 
\begin{equation}
b_i = c_i \qquad (1 \le i \le k)
\label{eq:boolean-cliff-result3}
\end{equation}
and hence the first $k$ rows of $b$ and $c$ are equal.

\eqref{eq:boolean-cliff-result1}, \eqref{eq:boolean-cliff-result2} and \eqref{eq:boolean-cliff-result3} determine $T$ and $b$ uniquely given the choice and ordering of generators.
\end{proof}

\section{Composition law for affine Clifford maps}
\label{app:clifford-composition}

\begin{proposition}[Composition law for affine Clifford maps]
Let $C_1,C_2$ be Clifford operators that act on computational basis states as an affine map on the bitstrings (i.e. $C_1\ket{q}=\ket{T_1 q \oplus b_1}$ and $C_2\ket{q}=\ket{T_2 q \oplus b_2}$).
Then the composition $C = C_2C_1$ is affine with tableau
\begin{equation}
\label{eq:clifford-theorem}
M=\begin{pmatrix} (T_2^{-1})^T(T_1^{-1})^{T}&0\\0&T_2T_1\end{pmatrix},
\qquad
s=\begin{pmatrix}T_1^{-1} ( b_1 \oplus T_2^{-1} b_2)\\0\end{pmatrix}.
\end{equation}
\end{proposition}

\begin{proof}
For any $\ket{q}$,
\begin{align}
C_1 \ket{q} = \ket{T_1 q \oplus b_1}, \qquad C_2 \ket{q} &= \ket{T_2 q \oplus b_2}, \label{eq:affine-defs}
\end{align}
Hence
\begin{align}
C \ket{q} &= C_2 C_1 \ket{q}\\
&= C_2 \ket{T_1 q \oplus b_1} \label{eq:compose-step1}\\
&= \ket{\,T_2\!\left(T_1 q \oplus b_1\right) \oplus b_2\,} \label{eq:compose-step2}\\
&= \ket{\, (T_2 T_1) q \oplus (T_2 b_1 \oplus b_2) \,}. \label{eq:compose-final}
\end{align}
Thus $T=T_2T_1$ and $b=b_2\oplus T_2 b_1$. Substituting these into the affine Clifford tableau form of Appendix~\ref{app:clifford-tableau}
($M=\mathrm{diag}\!\big((T^{-1})^T,\,T\big)$ and $s=\binom{T^{-1}b}{0}$) yields \eqref{eq:clifford-theorem}.
\end{proof}

\section{Bravyi--Kitaev as an affine Clifford}
\label{app:bk-affine-clifford}

We collect here the elementary facts needed to apply the affine-Clifford framework of Appendices~\ref{app:clifford-tableau}--\ref{app:clifford-composition} to the Bravyi--Kitaev (BK) mapping. Throughout, $n=2^x$ for some integer $x\geq 0$, and qubits/orbitals are indexed by $0,\dots,n-1$.

\begin{proposition}[Invertibility of the BK matrix]\label{lem:bk-invertible}
Let $T_n\in\mathbb F_2^{n\times n}$ be defined recursively by $T_1=(1)$ and
\begin{equation}\label{eq:bk-recursion-app}
T_{2n}=
\begin{pmatrix}
T_n & 0\\
A_n & T_n
\end{pmatrix},
\end{equation}
with $A_n\in\mathbb F_2^{n\times n}$ the matrix whose last row is $(1,\dots,1)$ and whose other entries are zero. Then $T_n$ is lower-triangular with unit diagonal, and in particular $T_n\in GL(n,\mathbb F_2)$.
\end{proposition}

\begin{proof}
Both properties are preserved by the recursion \eqref{eq:bk-recursion-app}: if $T_n$ is lower-triangular with unit diagonal, the block form shows the same for $T_{2n}$, since $A_n$ is strictly below the diagonal of $T_{2n}$ and the two diagonal blocks of $T_{2n}$ are $T_n$. The base case $T_1=(1)$ is trivial, and a lower-triangular matrix with unit diagonal over $\mathbb F_2$ has determinant $1$, hence is invertible.
\end{proof}

\begin{proposition}[BK is an affine Clifford]\label{prop:bk-affine-clifford}The Bravyi--Kitaev encoding is the unitary $C_{\mathrm{BK}}$ defined on Jordan--Wigner computational-basis states by
\begin{equation}\label{eq:bk-affine-action}
C_{\mathrm{BK}}\ket{f}=\ket{T_n f},\qquad f\in\mathbb F_2^n,
\end{equation}
with $T_n$ as in Lemma~\ref{lem:bk-invertible}~\cite{Bravyi2002,Seeley2012}. In particular $C_{\mathrm{BK}}$ is an affine Clifford map in the sense of Appendix~\ref{app:clifford-tableau} (with $T=T_n$ and $b=0$), and its tableau is
\begin{equation}\label{eq:bk-tableau-app}
M_{\mathrm{BK}}=
\begin{pmatrix}(T_n^{-1})^{\mathsf T}&0\\ 0&T_n\end{pmatrix},
\qquad s_{\mathrm{BK}}=0.
\end{equation}
\end{proposition}

\begin{proof}
That \eqref{eq:bk-affine-action} defines the BK encoding is the standard parity-tree construction of \cite{Bravyi2002,Seeley2012}: the rows of $T_n$ select, for each qubit, the subset of orbitals whose occupation parities are stored on that qubit, and the recursion \eqref{eq:bk-recursion-app} reproduces the parity tree. By Lemma~\ref{lem:bk-invertible}, $T_n\in GL(n,\mathbb F_2)$, so \eqref{eq:bk-affine-action} extends to a unitary on $(\mathbb C^2)^{\otimes n}$. Applying Proposition~\ref{prop:block-tableau} of Appendix~\ref{app:clifford-tableau} with $T=T_n$ and $b=0$ then gives \eqref{eq:bk-tableau-app}; in particular $C_{\mathrm{BK}}\in\mathrm{Cliff}_n$.
\end{proof}

\begin{proposition}[Tableau of SAE-CAS-BK]\label{cor:sae-cas-bk-tableau}
Let $C=C_2C_1$ be the SAE-CAS Clifford with tableau given by Eq.~\eqref{eq:clifford-theorem}, and let $C_{\mathrm{BK}}$ act on the active-space register as in Proposition~\ref{prop:bk-affine-clifford}. Then $C_{\mathrm{BK}}\,C_2\,C_1$ is an affine Clifford whose tableau is obtained by a single further application of the composition law of Appendix~\ref{app:clifford-composition}, with $(T_2,b_2)$ replaced by $(T_n,0)$ acting on the active subspace. In particular, the SAE-CAS-BK active-space qubit Hamiltonian is unitarily equivalent to the SAE-CAS one and shares its eigenspectrum.
\end{proposition}

The same construction applies verbatim to any other fermion-to-qubit encoding that can be expressed as an affine Clifford basis change of Jordan--Wigner---for instance the parity encoding~\cite{Seeley2012}---by replacing $T_n$ with the corresponding affine-Clifford matrix.

\section{Complete active space Hamiltonian}
\label{app:cas-second-quantised-proof}

\begin{proposition}[Complete active space Hamiltonian]
Define projectors ${P}_F = \bigotimes_{i\in F} \ket{1} \bra{1}_i,$ and ${P}_V = \bigotimes_{i\in V} \ket{0} \bra{0}_i$, and the complete active space projector 
\begin{equation}
     P
= {P}_F
\otimes I_A \otimes 
{P}_V,
\end{equation}
\qquad
where the sets $F$, $A$ and $V$ correspond respectively to the indices of the frozen-core, active and virtual spin orbitals. The complete active space Hamiltonian is 
\begin{equation}
    {H'}= P  H  P
\end{equation}, where $H$ is the molecular electronic structure Hamiltonian in Eq. \eqref{eq:hamiltonian}.

The constant term and the one-electron integrals of $H'$ can be written in terms of those of $H$ as in Eqs.~\eqref{eq:fcc-spin}--\eqref{eq:fc1ei-spin}, where $p,q\in A$, and all two--electron terms with indices entirely in $A$ are unchanged.
\end{proposition}

\begin{proof}
Write the second-quantized Hamiltonian as in Eq.~\eqref{eq:hamiltonian},
\begin{equation}
 H
= h_0
+ \sum_{pq} h_{pq}\, a_p^\dagger a_q
+ \tfrac12 \sum_{pqrs} g_{prqs}\, a_p^\dagger a_q^\dagger a_s a_r .
\label{eq:ham-as-used}
\end{equation}
By definition $ P$ projects onto configurations with frozen orbitals $F$ doubly (spin-)occupied and virtual orbitals $V$ empty, while acting as the identity on the active sector $A$. We have for frozen orbitals ($i \in F$)
\begin{equation} 
 P \, n_i =  n_i  P,\qquad
 P\,  a_i =0,\qquad
 a_a^\dagger \,  P =0, \
\label{eq:proj-core}
\end{equation} 
for virtual orbitals  ($a \in V$)
\begin{equation}  P \, n_a = 0, \qquad
 P\,  a_a^\dagger =0,\qquad
 a_a \,  P =0 \label{eq:proj-virt}
\end{equation}
and active operators commute through $ P$ ($p \in A$)
\begin{equation}
 P \, n_p =  n_p  P, \qquad  P\,  a_p^\dagger =  a_p^\dagger \,  P, \qquad  P\,  a_p =  a_p \,  P. \qquad 
\label{eq:proj-active}
\end{equation}

Because $ P^2= P$, we have $ H' P= P H P= H'$ and $ P H'= H'$, hence $[ H', P]=0$. Thus it suffices to work with $ H' P$; keeping a trailing $ P$ does not change the operator and simply records that all intermediate expressions are restricted to the CAS.

\medskip\noindent
\textbf{One-electron part.}
Consider $ P \Bigl(\sum_{p q} h_{p q}\,  a_p^\dagger  a_q\Bigr)  P$ and separate cases by orbital type.
\begin{enumerate}
\item If $p\in V$ or $q\in V$, the term vanishes by \eqref{eq:proj-virt}.
\item If $p\in F$ and $q\in A$, then
$ P  a_i^\dagger  a_p  P = ( a_i^\dagger  P)\,  a_p  P = 0$.
Similarly, if $p\in A$ and $q\in F$,
$ P  a_p^\dagger  a_i  P
=  a_p^\dagger ( P  a_i) P=0$,
using \eqref{eq:proj-active} and \eqref{eq:proj-core}.
\item If $p=q=i\in F$, then
$ P  a_i^\dagger  a_i  P
=  P  n_i  P =  P$.
\item If $p,q\in A$, the term is unchanged because active operators commute with $ P$.
\end{enumerate}
Therefore the projected one-electron operator equals
\begin{equation}
 P \Bigl(\sum_{p q} h_{p q}\,  a_p^\dagger  a_q\Bigr)  P
= \sum_{p,q\in A} h_{p q}\,  a_p^\dagger  a_q \,  P
+ \sum_{i\in F} h_{i i}\,  P .
\label{eq:1e-result}
\end{equation}

\medskip\noindent
\textbf{Two-electron part.}
Consider
$ P \Bigl(\tfrac12 \sum_{p q r s} g_{p q r s}\,  a_p^\dagger  a_q^\dagger  a_s  a_r\Bigr) P$.
Any term that contains a virtual index vanishes by \eqref{eq:proj-virt}.
Terms with exactly one or three frozen indices also vanish: after commuting active operators through $ P$ (by \eqref{eq:proj-active}), one is left with either $ P  a_i$ or $ a_i^\dagger  P$, which is zero by \eqref{eq:proj-core}.

\emph{(i) Two frozen and two active indices.}
Up to relabeling, the only nonzero monomials are
$ a_p^\dagger  a_i^\dagger  a_i  a_q$ and
$ a_p^\dagger  a_i^\dagger  a_q  a_i$
with $p,q\in A$ and $i\in F$. Using \eqref{eq:proj-active}--\eqref{eq:proj-active} and $\{ a_i^\dagger, a_q\}=0$ ($i\in F$, $q\in A$), we obtain
\begin{equation}
 P\,  a_p^\dagger  a_i^\dagger  a_i  a_q \,  P
=  a_p^\dagger ( P  n_i  P)  a_q
=  a_p^\dagger  a_q \,  P,
\qquad
 P\,  a_p^\dagger  a_i^\dagger  a_q  a_i \,  P
= -  a_p^\dagger  a_q ( P  n_i  P)
= -  a_p^\dagger  a_q \,  P.
\label{eq:2e-1e-reduction}
\end{equation}
Hence the two-electron part contributes the one-body shift
\begin{equation}
\sum_{p,q\in A}\sum_{i\in F}
\bigl(g_{p i q i}-g_{p i i q}\bigr)\,
 a_p^\dagger  a_q \,  P ,
\label{eq:fc1e-from-2e}
\end{equation}
which yields \eqref{eq:fc1ei-spin}.

\emph{(ii) Four frozen indices.}
The only nonzero monomials are, for $i,j\in F$,
$ a_i^\dagger  a_j^\dagger  a_j  a_i$ and
$ a_i^\dagger  a_j^\dagger  a_i  a_j$. Using fermionic anticommutation and \eqref{eq:proj-core},
\begin{equation}
 P\,  a_i^\dagger  a_j^\dagger  a_j  a_i \,  P
=  P\,  n_i  n_j \,  P =  P,
\qquad
 P\,  a_i^\dagger  a_j^\dagger  a_i  a_j \,  P
= -  P\,  n_i  n_j \,  P = -  P,
\label{eq:fourcore-evals}
\end{equation}
while the $i=j$ cases vanish identically. Thus the constant contributed by the two-electron part is
\begin{equation}
\tfrac12 \sum_{i,j\in F}\bigl(g_{i i j j}-g_{i j j i}\bigr),
\label{eq:const-from-2e}
\end{equation}
and together with the core trace from \eqref{eq:1e-result} this gives \eqref{eq:fcc-spin}.

\emph{(iii) All-active indices.}
If $p,q,r,s\in A$, then
\begin{equation}
 P\,  a_p^\dagger  a_q^\dagger  a_s  a_r\,  P
=  a_p^\dagger  a_q^\dagger  a_s  a_r\,  P,
\label{eq:allactive-unchanged}
\end{equation}
so the two-electron block over $A$ is unchanged.

\medskip
The resulting expressions \eqref{eq:1e-result}, \eqref{eq:fc1e-from-2e}, \eqref{eq:const-from-2e}, and \eqref{eq:allactive-unchanged} only have terms acting on active space orbital except for the projection ${P}$, which when restricted to the active space acts as the identity. Combining them we then obtain:
$ H'= P  H  P$ with
\eqref{eq:fcc-spin}--\eqref{eq:fc1ei-spin}, and all two-electron terms with indices in $A$ unchanged. This proves the proposition.
\end{proof}

\begin{proposition}[Hartree-Fock energy]
The Hartree--Fock energy is given by Eq.~\eqref{eq:ehf-spin}, where $\mathrm{occ}$ denotes the spin--orbitals occupied in the Hartree--Fock state. If the frozen-core is taken to be these occupied spin--orbitals and there are no active orbitals, the projected Hamiltonian $H'= P H P$ reduces to the scalar $E_{\mathrm{HF}}$.
\end{proposition}

\begin{proof}
Take $F=\mathrm{occ}$, $A=\varnothing$, and $V$ the virtuals: then $H'$ reduces to the constant term in Eq. \eqref{eq:fcc-spin}, which equals $E_{\mathrm{HF}}$ in Eq. \eqref{eq:ehf-spin}.
\end{proof}

\section{Equivalence of canonical complete-active space via $Z$-symmetry qubit removal}
\label{app:cas-equivalence}

\begin{proposition}[Projector identity]
\label{prop:cas-proj-identity}
Let $\Phi_{\mathrm{JW}}$ be the Jordan--Wigner map from $n$ spin-orbitals to $n$ qubits and
\begin{align}
 P&=\Bigl(\bigotimes_{i\in F}\ket{1}\!\bra{1}_i\Bigr)\otimes I_A \otimes
 \Bigl(\bigotimes_{a\in V}\ket{0}\!\bra{0}_a\Bigr),
\\
 Q&=\Bigl(\bigotimes_{i\in F}\tfrac{I-Z_i}{2}\Bigr)\otimes I_A \otimes
 \Bigl(\bigotimes_{a\in V}\tfrac{I+Z_a}{2}\Bigr).
\end{align}
Then $\Phi_{\mathrm{JW}}( P)= Q$ and, for the Hamiltonian $ H$ in Eq.~\eqref{eq:hamiltonian},
\begin{equation}
\label{eq:cas-commuting-identity}
\Phi_{\mathrm{JW}}( P  H  P)= Q\,\Phi_{\mathrm{JW}}( H)\, Q .
\end{equation}
\end{proposition}

\begin{proof}
In the JW basis, occupations are computational $Z$-eigenstates with
$\ket{1}\!\bra{1}=(I-Z)/2$ and $\ket{0}\!\bra{0}=(I+Z)/2$, so $\Phi_{\mathrm{JW}}( P)= Q$. The JW map is linear and multiplicative on products of creation/annihilation operators, hence $\Phi_{\mathrm{JW}}( P  H  P)=\Phi_{\mathrm{JW}}( P)\,\Phi_{\mathrm{JW}}( H)\,\Phi_{\mathrm{JW}}( P)$.
\end{proof}

\begin{proposition}[Equivalence with the canonical CAS Hamiltonian]
\label{prop:cas-equivalence}
Apply Eq.~\eqref{eq:cas-commuting-identity} and then set $Z_i=-1$ for $i\in F$ and $Z_a=+1$ for $a\in V$. After removing those fixed qubits, the resulting active-space qubit Hamiltonian is exactly the JW image of the canonical CAS Hamiltonian on $A$, with coefficients given by Eqs.~\eqref{eq:fcc-spin}--\eqref{eq:fc1ei-spin}.
\end{proposition}

\begin{proof}
By Proposition~\ref{prop:cas-proj-identity}, it suffices to analyze
$ Q\,\Phi_{\mathrm{JW}}( H)\, Q$ and then fix $Z$ on $F\cup V$.
We use the simple filter rules (for any $j\in F\cup V$):
\begin{equation}
\label{eq:Q-filter}
 Q\,X_j\, Q=0, \qquad  Q\,Y_j\, Q=0,\qquad
 Q\,Z_j\, Q=\begin{cases}- Q,& j\in F,\\ + Q,& j\in V.\end{cases}
\end{equation}
Pauli operators on active qubits commute through $ Q$.

\medskip\noindent
\textbf{One-electron part.}
Consider $\sum_{pq}h_{pq}\, a_p^\dagger  a_q$.
As in Appendix~\ref{app:cas-second-quantised-proof}, we split by orbital type and use \eqref{eq:Q-filter}.
\begin{enumerate}[label=(\roman*), leftmargin=14pt]
\item If $p\in V$ or $q\in V$, the JW image contains $X/Y$ on a virtual qubit, hence vanishes when sandwiched by $ Q$.
\item If exactly one of $p,q$ is frozen and the other is active, the term flips a frozen qubit and vanishes under $ Q$.
\item If $p=q=i\in F$, then $ a_i^\dagger  a_i= n_i$ maps to $(I-Z_i)/2$ and evaluates to $1$ under $ Q$, contributing $h_{ii}\,I$.
\item If $p,q\in A$, the term is unchanged (active operators commute through $ Q$).
\end{enumerate}
Therefore
\begin{equation}
\label{eq:1e-qubit-result}
 Q\,\Phi_{\mathrm{JW}}\!\Bigl(\sum_{pq}h_{pq} a_p^\dagger  a_q\Bigr)\, Q
\;=\; \sum_{p,q\in A} h_{pq}\, a_p^\dagger  a_q
\ + \ \sum_{i\in F} h_{ii}\,I,
\end{equation}
which matches the active block plus the core trace appearing in Eq.~\eqref{eq:fcc-spin}.

\medskip\noindent
\textbf{Two-electron part.}
Consider $\tfrac12\sum_{pqrs} g_{pqrs}\, a_p^\dagger  a_q^\dagger  a_s  a_r$.

If any index lies in $V$, the term vanishes under $ Q$. Indeed, whenever a
creation or annihilation operator acts on a virtual orbital, the JW image contains $X_a$ or $Y_a$
on that qubit and $ Q X_a  Q= Q Y_a  Q=0$. The only other possibility is
when the virtual appears as a number operator $ n_a$ (e.g. from Coulomb or exchange terms
with matched indices), but
\begin{equation}
 Q  n_a  Q
= Q \tfrac{I - Z_a}{2}  Q
=\tfrac{1-(+1)}{2}\, Q=0
\qquad (a\in V).
\end{equation}
Parity-string $Z$’s on virtual qubits evaluate to $+1$ and do not affect this conclusion.
Consequently, every two-electron term with at least one virtual index is zero.

Terms with exactly one or three frozen indices flip a frozen qubit and also vanish. The surviving cases mirror Appendix~\ref{app:cas-second-quantised-proof}:

\emph{(i) Two frozen and two active indices.}
Let $p,q\in A$ and $i\in F$. Up to relabeling, the only nonzero monomials are
$ a_p^\dagger  a_i^\dagger  a_i  a_q$ and
$ a_p^\dagger  a_i^\dagger  a_q  a_i$.
Under JW and \eqref{eq:Q-filter}, each reduces to a one-body operator on $A$:
\begin{equation}
 Q\,\Phi_{\mathrm{JW}}\!\bigl( a_p^\dagger  a_i^\dagger  a_i  a_q\bigr)\, Q
=  a_p^\dagger  a_q,
\qquad
 Q\,\Phi_{\mathrm{JW}}\!\bigl( a_p^\dagger  a_i^\dagger  a_q  a_i\bigr)\, Q
= -\, a_p^\dagger  a_q,
\end{equation}
so summing over $i\in F$ yields the shift
\begin{equation}
\label{eq:fc1e-from-2e-qubit}
\sum_{p,q\in A}\sum_{i\in F}\!\bigl(g_{p i q i}-g_{p i i q}\bigr)\, a_p^\dagger  a_q,
\end{equation}
which is exactly Eq.~\eqref{eq:fc1ei-spin}.

\emph{(ii) Four frozen indices.}
For $i,j\in F$, the two monomials
$ a_i^\dagger  a_j^\dagger  a_j  a_i$ and
$ a_i^\dagger  a_j^\dagger  a_i  a_j$
evaluate under $ Q$ to $+1$ and $-1$, respectively, giving the constant
\begin{equation}
\label{eq:const-from-2e-qubit}
\tfrac12\sum_{i,j\in F}\!\bigl(g_{iijj}-g_{ijji}\bigr),
\end{equation}
which matches the two-electron contribution in Eq.~\eqref{eq:fcc-spin}.

\emph{(iii) All-active indices.}
If $p,q,r,s\in A$, the term commutes through $ Q$ and is unchanged.

\medskip
Combining the one-electron result \eqref{eq:1e-qubit-result} with the two-electron contributions \eqref{eq:fc1e-from-2e-qubit} and \eqref{eq:const-from-2e-qubit} reproduces precisely the canonical CAS formulas \eqref{eq:fcc-spin}--\eqref{eq:fc1ei-spin}, with the two-electron block over $A$ unchanged. Since $Z$ on $F\cup V$ is fixed to $\mp 1$, those qubits carry no dynamics and are removed, leaving the JW image of the CAS Hamiltonian on the active orbitals.
\end{proof}

\section{Compatibility of symmetry-adapted encoding with CAS}
\label{app:compatibility}

\begin{proposition}[Compatibility of exact-symmetry basis changes with CAS qubit removal]
\label{prop:compat-sym-cas}
Let $H_q=\Phi_{\mathrm{JW}}( H)$ and
\[
 Q=\Bigl(\bigotimes_{i\in F}\tfrac{I-Z_i}{2}\Bigr)\otimes I_A \otimes
 \Bigl(\bigotimes_{a\in V}\tfrac{I+Z_a}{2}\Bigr)
\]
be the CAS qubit projector. Let $C_1$ be any affine Clifford implementing an exact-symmetry change of basis, and define $H_q':=C_1 H_q C_1^\dagger$ and $ Q':=C_1  Q\, C_1^\dagger$. Then
\begin{equation}
\label{eq:compat-equality}
 Q' \, H_q' \,  Q' \;=\; C_1 \bigl( Q\, H_q \,  Q\bigr) C_1^\dagger
\;=\; C_1\, \Phi_{\mathrm{JW}}( P  H  P)\, C_1^\dagger .
\end{equation}
After fixing the $Z$ eigenvalues on $F\cup V$ and removing those qubits, the two active-space Hamiltonians---obtained by applying $C_1$ before the CAS projection or after it---are related by an affine Clifford on the active qubits and are therefore isospectral.
\end{proposition}

\begin{proof}
Conjugation and associativity give
\[
 Q' H_q'  Q' 
= (C_1  Q C_1^\dagger)\,(C_1 H_q C_1^\dagger)\,(C_1  Q C_1^\dagger)
= C_1 \, Q H_q  Q\, C_1^\dagger .
\]
By Proposition~\ref{prop:cas-proj-identity}, $ Q H_q  Q=\Phi_{\mathrm{JW}}( P  H  P)$, proving \eqref{eq:compat-equality}. 

Because the symmetry qubits correspond to active spin-orbitals, $T_1$ and $b_1$ have the form:
\begin{equation}
T_1 = 
\begin{pmatrix}
I & 0 & 0\\
T_{AF} & T_{AA} & T_{AV}\\
0 & 0 & I\\
\end{pmatrix},\qquad b_1 =
\begin{pmatrix}0\\b_{A}\\0\end{pmatrix}
\end{equation}
So $C_1$ acts as the identity on the space of frozen-core and virtual qubits. Moreover the values of the frozen-core qubits are fixed to only ones and those of the virtual qubits are fixed to only zero, so the restriction of $C_1$ to the active qubits does not depend on them (it is given by $a = T_{AA} a \oplus b_A \oplus \bigoplus_i{{T_{AF}}_{;i}}$). Then the reduced active-space Hamiltonians differ by conjugation with an affine Clifford acting nontrivially only on $A$, and are therefore isospectral.

\end{proof}

\section{Convergence limitations for HE-SCA on JW-CAS}
\label{app:hesca-limitations}

In Section~\ref{sec:results} we report that the hardware-efficient
shifted-circular-alternating (HE-SCA) variational ansatz, when applied to
the unreduced JW-CAS register, does not converge to the CASCI
ground-state energy in the target symmetry sector for \ce{O2} (6,5) and
\ce{CO} (6,5) within the $0$--$200$ layer budget tested.  Two diagnostics
presented here identify the limitation as a combination of (i) leakage of
the variational state
into wrong $(N_\alpha,N_\beta)$ sectors, which the HE-SCA ansatz does not
constrain, and (ii) trapping of SLSQP in higher-energy local minima
within the unreduced 10-qubit register.  Neither mechanism is a
classical barren plateau: the gradient variance does not exhibit
exponential decay with layer count.  All numerical results in this
appendix are produced by the scripts in the \texttt{QuantumSymmetry}
repository with fixed RNG seeds (42 for the gradient scan,
0 for the sector-leakage runs) for reproducibility.

\paragraph{Gradient variance versus layer count rules out a barren plateau.}
A classical barren plateau is characterised by an exponential decay of the
gradient variance
$\mathrm{Var}_\theta\!\left[\partial\langle H\rangle/\partial \theta_k\right]$
with the number of variational layers \cite{McClean2018BarrenPlateaus}.  We sampled
$N_\mathrm{s}{=}50$ uniformly random parameter vectors
$\theta\in[0,2\pi)^{n_\mathrm{p}}$ for the JW-CAS HE-SCA ansatz at depths
$L\in\{1,5,25,50,100,150,200\}$ and computed the variance of
$\partial\langle H\rangle/\partial\theta_0$ over the ensemble via the
parameter-shift rule.  Figure~\ref{fig:hesca-mechanism}(a, b) shows the result
for the two five-orbital benchmarks.  The variance drops by roughly an order
of magnitude between $L{=}1$ and $L{=}5$ (factors of $26\times$ for \ce{O2}
and $12\times$ for \ce{CO}) and then \emph{saturates} in the range
$1$--$3\times 10^{-3}$ throughout $L\in[25,200]$, with no further
exponential decay---in fact it rises slightly at the deepest circuits.
Vanishing gradients are therefore \emph{not} the proximate cause of the
convergence limitation on these 10-qubit JW-CAS registers.

\paragraph{Sector leakage and local-minimum trapping.}
The HE-SCA ansatz is built from single-qubit $R_y$ rotations and CNOT
entanglers and so does not commute with the spin-up and spin-down number
operators
$\hat{N}_\alpha=\sum_{k}\tfrac{I-Z_{2k}}{2}$ and
$\hat{N}_\beta=\sum_{k}\tfrac{I-Z_{2k+1}}{2}$.  Without symmetry-based qubit removal
at the encoding level the variational manifold therefore spans all
$(N_\alpha,N_\beta)$ sectors of the active space.  To test whether this
sector leakage is the practical limitation we ran $N_\mathrm{t}{=}10$
SLSQP optimisations from independent uniformly random initialisations on
the JW-CAS HE-SCA ansatz at $L{=}10$ for both molecules, and on each
converged state measured the energy and the expectation values
$\langle \hat N_\alpha\rangle$, $\langle \hat N_\beta\rangle$,
$\langle \hat N\rangle$.

For \ce{O2}, the target sector is $(N_\alpha,N_\beta)=(4,2)$, the
$M_S{=}{+}1$ component of the triplet ${}^3\Sigma_g^-$ ground state; the
$M_S{=}{-}1$ partner $(2,4)$ is energetically degenerate with the target
by the spin invariance of $H$, and the $M_S{=}0$ triplet partner sits in
$(3,3)$ at the same energy.  Of the ten JW-CAS trials, 5 converge to a
triplet ground state within $\sim 1.5\,\mathrm{mHa}$ of the CASCI
energy---3 in $(4,2)$ and 2 in $(2,4)$.  Four further trials drift into
the $(3,3)$ sector and settle on higher-energy singlet eigenstates
between $+20$ and $+73\,\mathrm{mHa}$ above the triplet ground state;
the optimiser does not find the $M_S{=}0$ triplet partner from any of
the random initialisations sampled here, indicating that its basin of
attraction on the HE-SCA landscape is much smaller than the basins of
the excited singlets.  The remaining trial converges to a literal
sector superposition with $\langle N_\alpha\rangle=3.55$ and
$\langle N_\beta\rangle=2.44$---a state with no fixed particle number
(Fig.~\ref{fig:hesca-mechanism}(c), diamond).

For \ce{CO} (target sector $N_\alpha{=}N_\beta{=}3$, singlet
${}^1\Sigma^+$), only one trial leaks out of $(3,3)$ into $(2,4)$; the
other nine stay in the target sector.  Yet \emph{every} trial converges
to a local minimum between $+33$ and $+316\,\mathrm{mHa}$ above the
target FCI energy.  The distribution is outlier-dominated: 8/10 trials
cluster in a narrow band of $+33$ to $+60\,\mathrm{mHa}$ (median
$+47\,\mathrm{mHa}$), while two outliers at $+282$ and $+316\,\mathrm{mHa}$
pull the mean to $+97\,\mathrm{mHa}$.  No trial reaches chemical
accuracy.  For CO, then, the dominant limitation is not sector
leakage but local-minimum trapping within the correct sector.

SAE-CAS eliminates the first limitation by construction: removing the
two spin-parity generators fixes
$(N_\alpha\!\!\mod 2,\,N_\beta\!\!\mod 2)$ in the encoding, so the
variational state cannot leak between sectors of different parity
regardless of optimiser trajectory.  It also ameliorates the second
mode by reducing the register from $10$ qubits to $5$ (\ce{O2}) or
$6$ (\ce{CO}) and the HE-SCA parameter count by a factor of $\sim$5--10
at fixed $L$, shrinking the basin landscape that SLSQP must navigate.
The apples-to-apples test is decisive: at the same depth $L{=}10$ used
for the JW-CAS sector-leakage runs, all ten SAE-CAS trials converge to
the CASCI energy of the target sector to within $2\times 10^{-6}\,
\mathrm{mHa}$ for \ce{O2} and $5\times 10^{-4}\,\mathrm{mHa}$ for \ce{CO}
(Fig.~\ref{fig:hesca-mechanism}(c, d), triangles).  This is five to eight
orders of magnitude below the corresponding JW-CAS error distribution at
the same depth, and three to six orders of magnitude below chemical
accuracy.  The same SLSQP optimiser starting from the same uniform
$\theta_0\!\in[0,2\pi)^{n_\mathrm{p}}$ distribution reliably finds the
target ground state once the variational manifold is restricted to the
target symmetry sector by the encoding.

\begin{figure*}[tbp]
  \centering
  \includegraphics[width=0.95\textwidth]{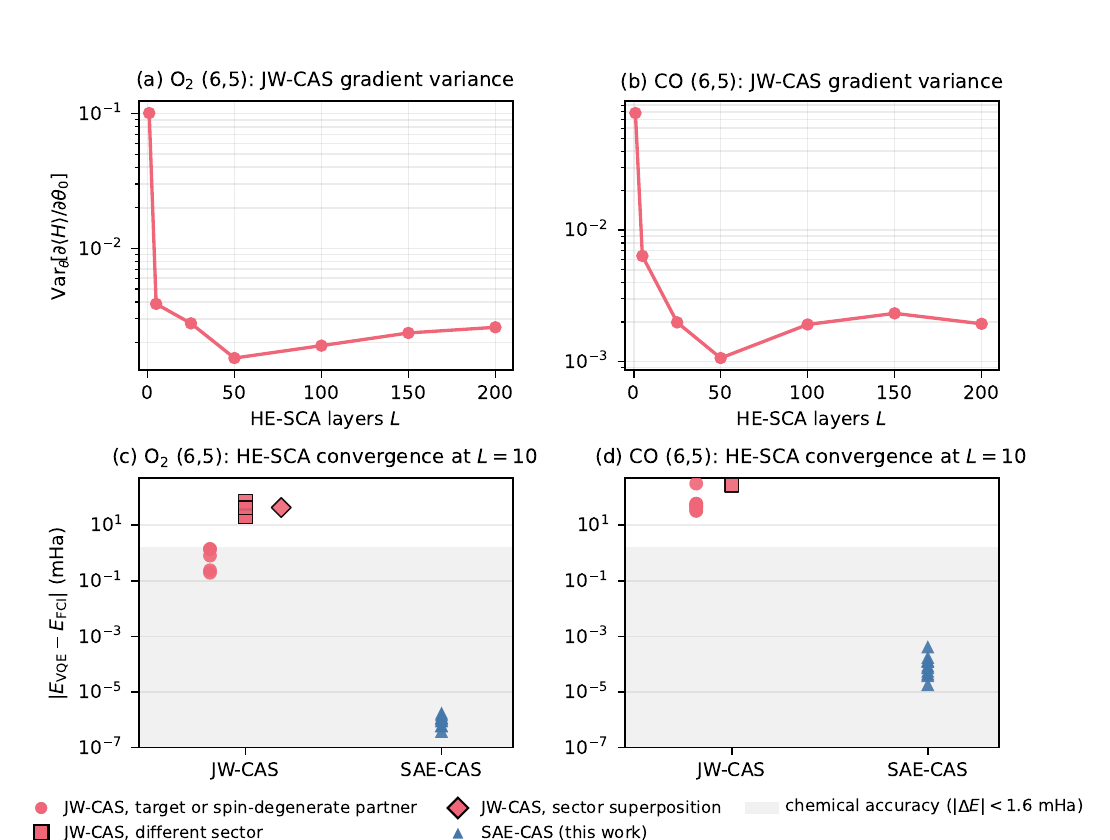}
  \caption{%
    Diagnostics for the HE-SCA convergence limitations on JW-CAS, side by
    side for \ce{O2} (6,5) (left column) and \ce{CO} (6,5) (right
    column).  \textbf{Top row, (a, b)}: gradient variance
    $\mathrm{Var}_\theta[\partial\langle H\rangle/\partial\theta_0]$ on
    the JW-CAS HE-SCA ansatz over 50 random $\theta$ samples, as a
    function of layer count $L$.  In both molecules the variance drops
    about two orders of magnitude from $L{=}1$ to $L{=}5$ and then
    saturates near $10^{-3}$; there is no exponential decay, ruling out
    a classical barren plateau.  \textbf{Bottom row, (c, d)}: per-trial
    absolute convergence error $|E_\mathrm{VQE}-E_\mathrm{FCI}|$ on a
    log scale, over 10 SLSQP runs from independent uniformly-random
    initialisations at the same depth $L{=}10$ for both JW-CAS and
    SAE-CAS.  Marker shapes encode the converged-state sector under
    JW-CAS: circles for the target $(N_\alpha,N_\beta)$ sector or its
    spin-degenerate partner (for \ce{O2}, both $(4,2)$ and $(2,4)$ are
    $M_S{=}\pm 1$ components of the triplet ground state); squares for
    trials that drifted to a different sector; the diamond in panel (c)
    marks the one \ce{O2} trial that converged to a literal sector
    superposition with non-integer
    $\langle N_\alpha\rangle,\langle N_\beta\rangle$.  Colour
    distinguishes the two encodings (JW-CAS in blue, SAE-CAS in
    orange).  SAE-CAS reaches the CASCI energy of the target sector to
    sub-microhartree precision in every trial of both molecules,
    placing all ten triangles deep inside the chemical-accuracy band.
  }
  \label{fig:hesca-mechanism}
\end{figure*}

\bibliography{references}

\end{document}